\documentclass[onefignum,onetabnum]{siamonline220329}

\usepackage{lipsum}
\usepackage{amsfonts}
\usepackage{graphicx}
\usepackage{epstopdf}
\usepackage{algorithmic}
\ifpdf
  \DeclareGraphicsExtensions{.eps,.pdf,.png,.jpg}
\else
  \DeclareGraphicsExtensions{.eps}
\fi

\usepackage{enumitem}
\setlist[enumerate]{leftmargin=.5in}
\setlist[itemize]{leftmargin=.5in}

\newsiamremark{remark}{Remark}
\newsiamremark{hypothesis}{Hypothesis}
\crefname{hypothesis}{Hypothesis}{Hypotheses}
\newsiamthm{claim}{Claim}
\newsiamthm{problem}{Problem}

\headers{Optimal Group Testing for Symmetric Distributions}{N. C. Landolfi and S. Lall}

\title{Optimal Dorfman Group Testing for Symmetric Distributions}%\thanks{Submitted to the editors DATE.\funding{funding}}}

\author{Nicholas C. Landolfi\thanks{Department of Computer Science, Stanford University, Stanford, CA 94305, USA
  (\email{lando@stanford.edu}).}
  \and Sanjay Lall\thanks{Department of Electrical Engineering, Stanford University, Stanford, CA 94305, USA
    (\email{lall@stanford.edu}).}
      }

\usepackage{amsopn}

\usepackage{amssymb}
\usepackage{mathtools}
\DeclarePairedDelimiter{\set}{\{}{\}}
\DeclarePairedDelimiterX{\Set}[2]{\{}{\}}{#1 \nonscript\;\delimsize\vert\nonscript\; #2}
\DeclareMathAlphabet{\mathbfsf}{\encodingdefault}{\sfdefault}{bx}{n}
\newcommand{\R}{\mathbb{R}}
\newcommand{\E}{\mathbb{E}}

\newcounter{edremcounter}
\def\mathword#1{\mathop{\mathrm{#1}}}
\newcommand{\Prob}{\mathword{Prob}}
\newcommand{\nnz}{\mathword{nnz}}
\newcommand{\num}[1]{\lvert#1\rvert}
\newcommand{\0}{\mathbf{0}}
\newcommand{\pow}[1]{\mathcal{P}(#1)}
\newcommand{\Disj}{\mathword{Disj}}
\newcommand{\IntParts}{\mathcal{L}}
\newcommand{\MultFns}{\mathcal{M}}

\newcommand{\MF}{\mathcal{M}}
\newcommand{\presup}[2]{{}^{#1}\!#2}
\newcommand{\Mstar}{M^\star}
\newcommand{\muiter}[1]{\mu^\star_{#1}}
\newcommand{\iiter}[1]{i^\star_{#1}}
\newcommand{\argmin}{\mathword{argmin}}

\newcommand{\extraproof}[1]{} % a proof we have commented out
\newcommand{\extraforproof}[1]{} % text only required or sensible if a proof is included
\newcommand{\extraproofalternative}[1]{#1} % alternative text for when a proof is omitted
\newcommand{\extraclarification}[1]{} % clarification which we have dropped for length

\let\bbl\Bigl

\let\bbr\Bigr

\ifpdf
\hypersetup{
  pdftitle={Optimal Dorfman Group Testing for Symmetric Distributions},
  pdfauthor={N. C. Landolfi and S. Lall}
}
\fi

\begin{document}

\maketitle

\begin{abstract}
  % WHAT
We study Dorfman's classical group testing protocol in a novel setting where individual specimen statuses are modeled as exchangeable random variables.
% WHY
We are motivated by infectious disease screening.
In that case, specimens which arrive together for testing often originate from the same community and so their statuses may exhibit positive correlation.
% MORE WHAT
Dorfman's protocol screens a population of $n$ specimens for a binary trait by partitioning it into non-overlapping groups, testing these, and only individually retesting the specimens of each positive group.
The partition is chosen to minimize the expected number of tests under a probabilistic model of specimen statuses.
We relax the typical assumption that these are independent and identically distributed and instead model them as exchangeable random variables.
In this case, their joint distribution is symmetric in the sense that it is invariant under permutations.
We give a characterization of such distributions in terms of a function $q$ where $q(h)$ is the marginal probability that any group of size $h$ tests negative.
We use this interpretable representation to show that the set partitioning problem arising in Dorfman's protocol can be reduced to an integer partitioning problem and efficiently solved.
We apply these tools to an empirical dataset from the COVID-19 pandemic.
The methodology helps explain the unexpectedly high empirical efficiency reported by the original investigators.%, as compared with that indicated by the classical theory.

\end{abstract}



\section{Introduction}
%%%% INTRO: OVERVIEW
Group testing is widely used to conserve resources while performing large-scale disease screening.
Logistical considerations often lead to the use of Dorfman's simple two-stage adaptive procedure in practice.
This protocol is usually based on probabilistic analyses of disease prevalence arising from models of specimen statuses as mutually independent random variables.
In this paper, we generalize and study the case in which the statuses are modeled as exchangeable, but not necessarily independent, random variables.

%%%% INTRO: GROUP TESTING
Given a population of $n$ specimens to screen for a binary trait, the group testing framework allows for several specimens to be pooled and tested together as a group.
The group tests positive if any of its individual specimens is positive.
The group tests negative if, and only if, all of its specimens are negative.
Numerous protocols using this testing capability have been proposed, of which Dorfman's two-stage adaptive procedure is the earliest, simplest, and most widely used.
In this protocol, the population is partitioned into non-overlapping groups and these are tested in the first stage.
If a group of size $h > 1$ tests negative, each of its $h$ specimens is immediately determined negative and $h-1$ tests are saved.
If a group tests positive, each of its specimens is retested individually in the second stage and determined according to the outcome of its individual test.
The key question is how to partition the specimens.

%%%% INTRO: *PROBABILISTIC* GROUP TESTING AND TYPICAL ASSUMPTIONS
Since tests are saved only when a group tests negative and these group test outcomes depend on the distribution and prevalence of positive specimens, a standard approach specifies a probabilistic model of specimen statuses and finds a partition to minimize the expected number of tests used.
In general, both specifying the model and finding the partition are difficult.
The first requires a parameterization, and the second a computation, which grows exponentially in the number of specimens to be tested.
Historically, this complexity has been avoided via simple probabilistic models arising from a strong assumption of independence.

%%%% INTRO: WHY SYMMETRY? / EXCHANGEABILITY
It is desirable from both a theoretical and practical point of view to alleviate the independence assumption.
From a theoretical point of view, it is interesting to consider how one might efficiently find partitions for more complicated distributions.
From a practical point of view, it is natural to suppose that the statuses of specimens arriving together for testing may be correlated because they originate from the same family, living place, or workplace and the disease is contagious.
Indeed, a recent large-scale study cited this phenomenon when explaining the failure of current theoretical tools to predict observed empirical test savings \cite{barak2021lessons}.
It is a pleasant surprise, therefore, that one can model specimen statuses as exchangeable while maintaining interpretability of the probabilistic model and tractability of the computation.

\subsection{Contributions}

Before listing contributions, we roughly frame the mathematical problem that we address in this paper.
For details, see \Cref{section:problem}. % and \Cref{prob:mintests}.
We are given $n$ binary random variables $x_1, \dots, x_n$ with some underlying joint probability distribution.
We seek a partition $G$ of the set $\set{1, \dots, n}$ to minimize $\sum_{H \in G} f(H)$ where
\[
  f(H) = \begin{cases}
    1 & \text{if } \num{H} = 1 \\
    1 + \num{H}\Prob(\max_{i \in H} x_i = 1) & \text{otherwise}
  \end{cases}
\]
This approach minimizes the expected number of tests used by Dorfman’s procedure.

When individual specimen statuses are modeled as exchangeable random variables their joint distribution is symmetric in the sense that it is invariant under permutations of its arguments.
Our first contribution is to characterize such a symmetric distribution in terms of a function $q$, where $q(h)$ is the probability that a group of size $h$ tests negative.
The representation $q$ is key to finding a partition to minimize the expected number of tests.

Our second contribution is to use this characterization, along with a natural reduction of additive and symmetric set partitioning problems to additive integer partitioning problems, to show how to efficiently compute optimal partitions for exchangeable statuses.
In contrast to additive \emph{set} partitioning problems, additive \emph{integer} partitioning problems are tractable and several efficient algorithms are known for their solution.
For details, see \Cref{section:problems}.

Lastly, we apply these tools to an empirical dataset from the COVID-19 pandemic.
The data we use indicate empirical efficiency exceeding that predicted by the classical theory, which models statuses as independent and identically distributed.
Our tools partially explain this empirical efficiency and also indicate a different and more efficient partition than that used by the original investigators.
We make our numerical implementation available \cite{landolfi2023symgt}.

In summary, we study Dorfman's two-stage adaptive group testing procedure for the case of exchangeable specimen statuses.
Our three contributions are:
\begin{enumerate}
  \item a characterization of symmetric distributions over binary outcomes
  \item a method to efficiently find optimal testing partitions under such distributions
  \item a numerical experiment applying these tools to an empirical COVID-19 dataset
\end{enumerate}

\paragraph{Outline}
In the following two subsections we give further background and introduce our notation.
In \Cref{section:problem}, we formalize Dorfman's two-stage adaptive group testing protocol.
In \Cref{section:symmetricdistributions}, we discuss and characterize symmetric distributions.
In \Cref{section:problems}, we study the structure of symmetric and additive set partitioning problems and present tools for their solution.
In \Cref{section:experiment}, we numerically apply these tools to a COVID-19 dataset.
In \Cref{section:prior_work}, we review prior work on group testing.
We conclude in \Cref{section:conclusion} with some future directions.

\subsection{Background}
We provide background on group testing, probabilistic symmetries and partitioning problems.
Each is a highly developed field with an extensive literature.

\subsubsection{Group testing}
In 1943, Dorfman initiated the study of group testing, also called \textit{pooled testing}, by proposing his original methodology for disease screening \cite{dorfman1943detection}.
The field has grown considerably since.
Today, it can be distinguished along several axes.
We outline these further in \Cref{section:prior_work} below.
Here, we briefly characterize the setting of this paper.

\paragraph{Our setting} 
We study Dorfman's two-stage, adaptive procedure in the probabilistic, finite-population setting with binary specimen statuses and binary, noiseless, unconstrained tests.
The central novelty is in modeling the specimen statuses as exchangeable random variables.
We are aware of only one other article studying restricted forms of exchangeability \cite{lendle2012group}.
We also emphasize that we use the term symmetric, see \Cref{section:symmetricdistributions}, to describe the joint distribution of the statuses and not, as others have done \cite{sobel1971symmetric,du2000combinatorial}, to describe the testing model.

The key prior work for situating our contribution is Dorfman's original article \cite{dorfman1943detection} and Hwang's follow-up \cite{hwang1975generalized}.
Both considered Dorfman's \emph{adaptive} two-stage procedure in a \emph{probabilistic} setting, with \emph{noiseless} \emph{binary} test results.
Hwang moved from Dorfman's \emph{infinite} population setting to a \emph{finite} population setting and generalized Dorfman's probabilistic model to allow for specimen-specific positive status probabilities.
Our work generalizes Dorfman's in a similar but parallel way. 
Rather than dropping the \emph{identically distributed} assumption as Hwang does, we drop the \emph{independence} assumption.
We visualize these steps in \Cref{figure:assumptions}.

\begin{figure}[htbp]\label{figure:assumptions}
  \centering
  \includegraphics[width=\textwidth]{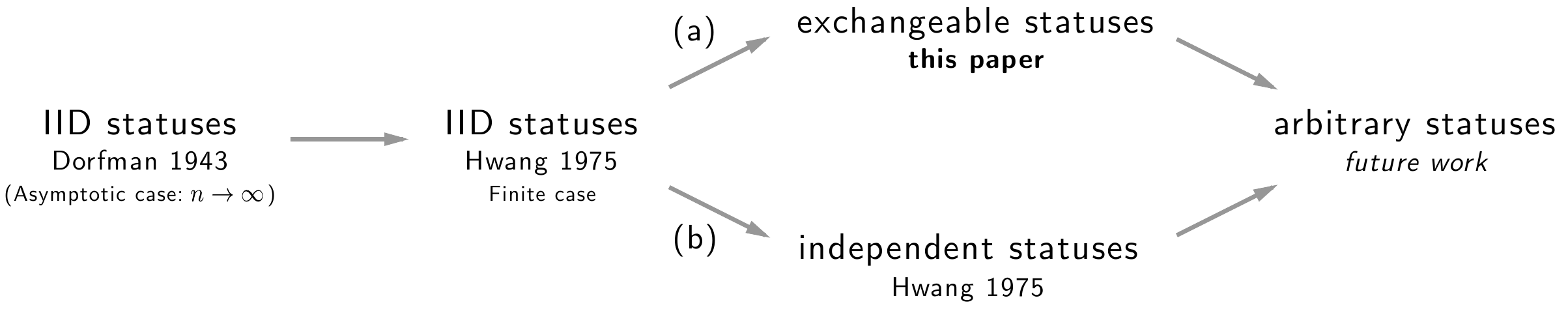}
  \vspace{-2.0em}
  \caption{
    Assumptions for Dorfman's two-stage adaptive group testing procedure with noiseless binary test results.
    (a) drops the independence assumption whereas (b) drops the identically distributed assumption.
  }
\end{figure}

\subsubsection{Probabilistic symmetries}
Exchangeable random variables fit within the broad study of probabilistic, or distributional, symmetries \cite{kallenberg2005probabilistic}.
An \emph{exchangeable} sequence of random variables is one whose joint distribution is \emph{invariant} under permutations \cite{aldous1985exchangeability}.
This condition is strictly \emph{weaker} than assuming that the sequence is IID.
Although we focus on \emph{finite} sequences, the concept first gained prominence when applied to \emph{infinite} sequences.

\paragraph{Infinite exchangeable sequences}
These are associated with an influential and well-known theorem of de Finetti, subsequently generalized by Hewitt and Savage \cite{hewitt1955symmetric}.
See \cite{kallenberg2005probabilistic} for a modern treatment.
Roughly speaking, \emph{de Finetti's theorem} says that the joint distribution of every \emph{infinite} exchangeable sequence can be expressed as a mixture of IID joint distributions \cite{definetti1931funzione,kallenberg2005probabilistic}.
Conversely, \emph{any} such \emph{IID mix} is exchangeable.
Permutation invariance, therefore, is \emph{characterized} by a representation which can be interpreted as a \emph{prior} distribution over the parameters of an infinite IID model.
Freedman \cite{freedman1995some} gives an informal discussion of this result and its relevance to the Bayesian, or subjective, interpretation of probability.

\paragraph{Finite exchangeable sequences}
de Finetti's characterization fails for \emph{finite} sequences \cite{diaconis1977finite}.
A more delicate treatment can be given, however, which approximates his result \cite{diaconis1980finite}.
In the sequel, we call the distributions of \emph{finite} exchangeable sequences \emph{symmetric}.
It is well-known that such distributions are mixtures of distributions of \emph{urn sequences} \cite{kallenberg2005probabilistic}. 
See \Cref{proposition:finitefinetti} below for a precise statement.
Our contribution is a separate characterization of symmetric distributions over \emph{binary} domains.
See \Cref{theorem:zeromarginals}.
We are not aware of this specialized result appearing explicitly in prior literature.
In the context of Dorfman's procedure, it is the key object which aids interpretation of the probabilistic model.

\subsubsection{Partitioning problems}
In the sequel, we encounter both \emph{set} and \emph{integer} partitioning problems.
Each has been extensively studied \cite{balas1976set,hwang2011partitions,engel2014optimal,onn2015some} and can be viewed as a particular combinatorial optimization problem \cite{lawler1976combinatorial,korte2018combinatorial}.
We mention that \emph{neither} is exactly the well-known ``partition" problem described by Karp in his classic paper \cite{karp1972reducibility,garey1979computers}.

\paragraph{Sets} 
In \emph{set} partitioning problems, we seek a partition of a finite set to minimize a given real-valued objective function.
Such a partition is sometimes called \emph{unlabeled} to distinguish it from an \emph{allocation}, which has a prespecified number of elements \cite{hwang2011partitions}.
For the many applications of these problems, see \cite{balas1976set} and \cite{hwang2011partitions}.
Dorfman's procedure partially motivated one historical line of work \cite{hwang1975generalized,hwang1981optimal,hwang1985optimal,hwang2011partitions}.
The basic difficulty is that the number of partitions of a finite set of size $n$, the so-called $n$th \emph{Bell number} \cite{becker1948arithmetic,rota1964number}, grows quickly with $n$.
Still, these problems have standard integer linear programming formulations when the objective is additive \cite{balas1976set,schrijver1998theory,schrijver2004combinatorial} and
 other structured objectives have been studied \cite{hwang1981optimal,anily1991structured,hwang2011partitions,lamarche-perrin2014generic}.

\paragraph{Integers} 
In \emph{integer} partitioning problems, we seek a partition of a positive integer \cite{hardy1979introduction,andrews1984theory} to minimize a given real-valued objective function.
As with set partitioning problems, the basic difficulty is that the number of partitions of a positive integer $n$ is large, even for moderate $n$.
We know of two outstanding articles which study these problems under additive objectives \cite{engel2014optimal,onn2015some}.
We discuss these in \Cref{subsection:solutions}.
Integer partitioning arises in this paper from a set partitioning problem whose objective is \emph{symmetric}.
See \Cref{section:problems}.
Although this reduction is natural, we are not aware of prior work explicitly making the connection.
Detecting and exploiting symmetry is an active area of research in combinatorial optimization \cite{margot2010symmetry,pfetsch2019computational}.

\subsection{Preliminaries}

% note: we use the | for restrictions to distinguish it from marginals
For finite sets $P$ and $D$, let $D^P$ denote the set of functions mapping $P$ to $D$.
Given $z:P \to D$ and $H \subset P$, denote the restriction of $z$ to
$H$ by $z_{|H}:H \to D$.  
Denote the constant zero function with any domain by $\0$.  
For any finite set $P$ and $u \in D^P$ with $0\in D$, define $\nnz(u) =
\num{\Set{i \in P}{u(i) \neq 0}}$, the number of points at which $u$ is nonzero.
Denote the empty set by $\varnothing$.
Denote the union of a set of sets $E$ by $\cup E$.

For $f: D^J \to C$, given $g: J \to H$, define $f^g:  D^H \to C$ via
% \[
%   f^g(x) = f(x \circ g) \quad \text{for all } x \in D^H
% \]
$f^g(x) = f(x \circ g)$ for all $x \in D^H$.
For $z: P \to D$, given $d \in D$, define $z^{-1}(d)  = \Set{i \in P}{z(i) = d}$, the preimage of $d$ under $z$.
% note: we use superscript before to avoid collision with $f^g$
% note: we use \presup command for better spacing; see definition in ../shared.tex
Given a set $F$ of subsets of a set $P$ and a function $h: P \to P$, define $\presup{h}{F}$ by
% \[
%   \presup{h}{F} = \Set{\Set{h(i)}{i \in H}}{H \in F}
% \]
$\presup{h}{F} = \Set{\Set{h(i)}{i \in H}}{H \in F}$.
Hence $\presup{h}{F}$ is the set of images under $h$ of the sets in $F$.

\paragraph{Probability}

Given a distribution $p: D^P \to [0,1]$, the probability of an event $A \subset D^P$ is $\sum_{z \in A} p(z)$.
We denote it by $\Prob(A)$ when $p$ is clear from context.
Given a set $H \subset P$, the \emph{marginal} of $p$ \emph{over} $H$ is the function $p_H: D^H \to [0,1]$ defined by $p_H(u) = \sum_{z \mid z_{|H} = u} p(z)$.

If $r: D^P \to [0,1]$ is also a distribution, the \textit{cross-entropy} $H(r, p)$ of $p$ relative to $r$ is $-\sum_{z \in D^P} r(z) \log p(z)$ and the entropy $H(r)$ of $r$ is $-\sum_{z \in D^P} r(z) \log r(z)$ as usual.
The \textit{Kullback-Leibler divergence} $d_{kl}(r, p)$ of $p$ relative to $r$ is defined as usual by $d_{kl}(r, p) = H(r, p) - H(r)$.
The empirical distribution $\hat{p}: D^P \to [0,1]$ of a dataset $z^1, \dots, z^m$ in $D^P$ is defined as usual by $\hat{p}(z) = (1/m)\num{\Set{i \in \set{1, \dots, m}}{z^i = z}}$.

\paragraph{Set powers}
Given a set $S$, the \emph{power set} $\pow{S}$ of $S$ is the set of all subsets of $S$.
The power set \emph{of} $\pow{S}$ is the set of all \emph{sets} of subsets of $S$.
We denote the nonempty elements of this set whose members are nonempty and disjoint by $\Disj(S)$.

\paragraph{Set partitions}

A \emph{partition} $F = \set{F_1, \dots, F_r}$ of a \emph{set} $S$ is a set of nonempty, pairwise disjoint subsets of $S$ whose union is $S$.
That is, $F_i \cap F_j = \varnothing$ whenever $i \neq j$ and $\cup_{i = 1}^{r} F_i = S$.
Given a set $P$, cost function $J: \Disj(P) \to \R$ and any nonempty $S \subset P$, we call a \emph{partition} $F^\star$ of $S$ \emph{optimal} for $S$ under $J$ if $J(F^\star) \leq J(F)$ for all partitions $F$ of $S$.

\paragraph{Integer partitions}

A \emph{partition} $\lambda = (\lambda_1, \dots, \lambda_r)$ of the positive \emph{integer} $m$ is a nonincreasing finite sequence of positive integers whose sum is $m$ \cite{hardy1979introduction,andrews1984theory}.
The terms $\lambda_i$ are called \emph{parts}.
The \emph{multiplicity} of an integer in $\lambda$ is the number of times it appears as a part \cite{macdonald1998symmetric}.
We associate to $\lambda$ a \emph{multiplicity function} $\mu$ so that $\mu(h)$ is the multiplicity of the integer $h$ in $\lambda$.
We denote the integer partitions of $m$ by $\IntParts(m)$ and the corresponding multiplicity functions by $\MultFns(m)$.
There is a bijection between $\MultFns(m)$ and $\IntParts(m)$.
We denote the set $\cup_{i = 1}^{m} \MultFns(i)$ by $\MultFns(1, \dots, m)$.

\paragraph{Set and integer partitions}
Given a partition $F$ of a nonempty \textit{set} $S$, we can construct a partition $\lambda_F$ of the positive \emph{integer} $m = \num{S}$.
The parts of $\lambda_F$ are the sizes of the elements of $F$, in nonincreasing order as usual.
This integer partition $\lambda_F$ has a multiplicity function $\mu_F$, where $\mu_F(h)$ is the number of parts of size $h$ in $F$. 
We also call $\mu_F$ the \emph{multiplicity function} of the set partition $F$.
Any $F \in \Disj(S)$ is a partition of the set $\cup F \subset S$, and so has corresponding integer partition $\lambda_F \in \IntParts(k)$ and multiplicity function $\mu_F \in \MultFns(k)$ where $k = \num{\cup F}$.
Given $F, G \in \Disj(S)$, we call $F$ and $G$ \emph{multiplicity equivalent} if $\mu_F = \mu_G$. 
This holds if and only if $\lambda_F = \lambda_G$.
Note that possibly $\cup F \neq \cup G$.
If $F \cap G = \varnothing$, then $\mu_{F \cup G} = \mu_F + \mu_G$.

\section{Problem formulation}
\label{section:problem}
We have a \emph{population} $P$ of $n$ specimens to test for a binary trait.
A specimen is either \emph{negative} or \emph{positive}, which we denote by $0$ and $1$, respectively.
We model these $n$ statuses as random variables $\{x_i\}_{i \in P}$ with distribution $p: \set{0,1}^P \to [0, 1]$.
Here each outcome is a binary function on $P$ and $p$ assigns a probability to each outcome.

\subsection{Group testing}
We determine the statuses via testing.
We may test several specimens together and observe that either (a) all of the specimens are negative or (b) at least one of the specimens is positive.
A \emph{group} is a nonempty subset $H \subset P$.
Its status is defined to be $S_H(x) = \max_{i\in H} x_i$.
We say that the group $H$ tests negative if and only if $S_H(x) = 0$.
In other words, all of its members are negative.
A group $H$ tests positive means its status $S_H(x) = 1$.
There is no noise in the observed outcomes of individual or group tests.

\subsection{Dorfman's adaptive procedure}
Dorfman \cite{dorfman1943detection} proposed determining specimen statuses via a two-stage procedure.
The population is first partitioned into groups and these are tested.
If a group tests negative, each specimen in the group is determined negative.
If a group tests positive, each specimen in the group is retested individually, and is determined positive or negative depending on the result of its individual test.

Given the statuses $x$ and a group $H \subset P$, the number of tests required to determine the status of every specimen in $H$ is
\begin{equation}\label{eq:numtests}
  T_H(x) = \begin{cases}
    1 & \text{if } \num{H} = 1 \\
    1 + \num{H} S_H(x) & \text{otherwise}
  \end{cases}
\end{equation}
The mean of this random variable is then
\begin{equation}\label{eq:expectedtests}
  \E T_H(x) = \begin{cases}
    1 & \text{if } \num{H} = 1 \\
    1 + \num{H}\Prob(S_H(x) = 1) & \text{otherwise}
  \end{cases}
\end{equation}
The first case of \cref{eq:numtests} records that a group with one member requires only one test.
Otherwise, a group $H$ of size two or more requires one group test and possibly $\num{H}$ additional individual tests.
These additional tests are required only if the group status is positive.

  Dorfman's procedure may be applied to any nonempty \emph{subpopulation} $S \subset P$.
  Given a \emph{partition} $F$ of $S$, the number of tests used to determine the status of every specimen in $S$ is $C(F,x) = \sum_{H \in F} T_H(x)$, and its expectation is
\begin{equation}\label{eq:totalexpectedtests}
  \textstyle
  \E C(F,x) = \sum_{H \in F} \E T_H(x)\end{equation}
which is the sum of the expected number of tests needed for each group in $F$.
To determine the status of every specimen in the population, one is interested in a partition of $P$.

A \emph{pooling} of $P$ is a partition $G = \set{G_1, \dots, G_r}$ of $P$, where each group $G_i \subset P$.
A natural cost for a pooling $G$ is the expected number of tests $\E C(G, x)$ it uses.
A natural measure of its \emph{efficiency} is $n/\E C(G, x)$.
Here $n$ is the cost of testing each specimen individually.

\subsection{Minimizing expected number of tests}
It is natural to seek a partition which minimizes the expected number of tests required to determine the status of all specimens.
Or, equivalently, to seek a partition which maximizes efficiency.
\begin{problem}\label{prob:mintests}
  Given a distribution $p: \set{0, 1}^P \to [0, 1]$, find a partition $G$ of the population $P$ to minimize the expected number of tests $\E C(G, x)$.
\end{problem}
We are interested, therefore, in solving a set partitioning problem.
Without further assumptions, the problem is computationally challenging because of the large number of parameters required to specify $p$ and the large number of partitions.
Consequently, one is interested in particular distribution classes with succinct representations and efficient algorithms.

Hwang \cite{hwang1975generalized} showed that if specimen statuses are modeled as \emph{independent} random variables, then $p$ is determined by $n$ real parameters and \Cref{prob:mintests} can be efficiently solved.
We show herein that similar results hold if the statuses are instead modeled as \emph{exchangeable}.

\section{Symmetric distributions}
\label{section:symmetricdistributions}
Given a permutation $g$ on $P$, we can apply it to outcomes $x\in\{0,1\}^P$ in the natural way via composition to give
$x \circ g$. This also induces a corresponding rearrangement
$p^g$
of a distribution $p$ on $\{0,1\}^P$.
Call $p$ \emph{symmetric} if $p = p^g$ for all permutations $g$ on $P$.
The statement that $p$ is symmetric is equivalent to the probabilistic statement that the individual specimen status random variables are \emph{exchangeable} \cite{diaconis1977finite, kingman1978uses, aldous1985exchangeability, kallenberg2005probabilistic}.

\begin{remark}\label{rmk:nnz}
  Given $x$ and $y$ in $\set{0, 1}^P$, relate $x \sim y$ if there is a permutation $g$ so that $x = y \circ g$.
  The relation $\sim$ is an equivalence relation, and we have $x\sim y$
  if and only if $\nnz(x) = \nnz(y)$. The resulting equivalence classes
  are the sets of functions in $\set{0, 1}^P$ having the same number of
  nonzero values. A distribution $p$ is symmetric if and only if it is
  invariant on the equivalence classes, that is $p(x) = p(y)$ whenever
  $x \sim y$.
  In other words, $p$ is symmetric if and only if there exists a function $w: \set{0, \dots, n} \to [0,1]$ such that $p(x) = w(\nnz(x))$ for all $x \in \{0,1\}^P$.
  The value $w(k)$ gives the probability of a \textit{particular} outcome with $k$ nonzero values, for $k = 0, \dots, n$. 
  Since there are ${n \choose k}$ such outcomes with $k$ nonzero values, the probability of the \emph{event} of observing an outcome with $k$ nonzero values is ${n \choose k}w(k)$.
\end{remark}

\subsection{Examples of symmetric distributions}
\label{subsection:examples}
Here, we include a few notable examples.

\subsubsection{IID random variables have symmetric distributions}

The joint \emph{IID distribution} of a set $\set{x_i}_{i \in P}$ of independent and identically distributed random variables is symmetric.
% Given any distribution $b$ on $\{0,1\}$, we can define a symmetric distribution $p$ on $\{0,1\}^P$ by $p(x) = \prod_{i \in P} b(x(i))$ for all $x \in \{0,1\}^P$.
% $p$ is called an \emph{IID distribution}.
In the context of group testing, the value $\rho = \E \sum_{i = 1}^{n} x_i/n$ is called the \emph{prevalence rate} \cite{dorfman1943detection}.
% Here $\rho = b(1)$.
% We can express $p$ in terms of $\rho$ as $p(x) = (1 - \rho)^{n - \nnz(x)} \rho^{\nnz(x)}$ for all $x \in \{0,1\}^P$.
% This expression exhibits the symmetry of $p$ because we have written $p(x)$ as a function of $\nnz(x)$.
% See \Cref{rmk:nnz}.

\subsubsection{Mixtures of symmetric distributions are symmetric}

The set of symmetric distributions is convex.
% As usual, we call a convex combination of distributions a \emph{mixture}.
%
% A simple example is a mixture of two distributions.
% Given symmetric distributions $r$ and $s$ along with a mixing parameter $\mu$ in $[0,1]$, define the symmetric distribution $p$ on $\{0,1\}^P$ by $p(x) = (1-\mu) r(x) + \mu s(x)$ for all $x \in \{0,1\}^P$. % if commenting in this line, comment out the next line
For example, the mixture distribution $p$ defined by $p(x) = (1-\mu) r(x) + \mu s(x)$ for all $x \in \{0,1\}^P$ is symmetric, where $\mu$ is in $[0,1]$ and  $r,s$ are symmetric distributions.
% We may interpret $p$ as modeling statuses which depend on an \emph{unobserved} event with probability $\mu$ of occurring.
% In the context of group testing, one might call $p$ an \emph{exposure distribution} with two levels.
% If $r$ and $s$ have prevalence rates $\rho_r$ and $\rho_s$, respectively, then $p$ has rate $(1-\mu) \rho_r + \mu \rho_s$.
% If $\rho_s > \rho_r$, the unobserved \emph{exposure} event \emph{increases} the prevalence.
% The generalization to $\ell$ levels is straightforward.

Mixtures of IID distributions provide examples of symmetric distributions that model random variables which are \emph{not} independent.
For an extreme but easy to see case, suppose $r$ and $s$ of the previous paragraph have prevalence rates $0$ and $1$ respectively, and $\mu = 1/2$. 
Let $i$ and $j$ be distinct elements of $P$.
The probability of the event $\Set{x}{x(i) = 1}$ is $1/2$.
The \emph{conditional} probability of this event \emph{given} the event $\Set{x}{x(j) = 1}$ is $1$.
Consequently the two events are dependent, and so the random variables $x_i$ and $x_j$ are \emph{not} independent.
In other words, for such a distribution, if one specimen is positive then so are all the others.

\subsubsection{Simple random sampling is symmetric}

Suppose we draw all 0-1-labeled balls from an urn \emph{without} replacement, and record the labels $x_1, \dots, x_n$.
With $P = \set{1, \dots, n}$, the set $\set{x_i}_{i \in P}$ is exchangeable.
If $k \leq n$ of the balls were marked 1, then this set has symmetric distribution $r_k$ on $\{0,1\}^P$ defined by $r_k(x) = 1/{n \choose k}$ if $\nnz(x) = k$ and 0 otherwise.
A classic result \cite{definetti1972probability,kendall1967finite,ericson1976bayesian,diaconis1977finite} says that every symmetric distribution is a mixture of $r_0, \dots, r_n$.

\begin{proposition}\label{proposition:finitefinetti}
  Suppose $p$ is a distribution on $\set{0,1}^P$.
  Then $p$ is symmetric if and only if there exists a function $\alpha: \set{0, \dots, n} \to [0,1]$ such that $\sum_{k = 0}^{n} \alpha(i) = 1$ and $p(x) = \sum_{k = 0}^{n} \alpha(k) r_k(x)$ for all $x \in \{0,1\}^P$.
\end{proposition}

This fact is easy to see using \Cref{rmk:nnz}.
The function $\alpha$ is related to the function $w$ by $\alpha(k) = {n \choose k}w(k)$ for $i = 0, \dots, n$.
Hence, $\alpha(k)$ is the probability of observing an outcome with $k$ nonzero values.
Both $w$ of \Cref{rmk:nnz} and $\alpha$ of \Cref{proposition:finitefinetti} are \emph{representations} of a symmetric distribution.
We use $\alpha$ for visualization purposes in \Cref{figure:models_analysis} below.

\subsubsection{Shuffling symmetrizes}
\label{subsubsection:shuffling}

Given any distribution $t$ on $\{0,1\}^P$, \emph{not} necessarily symmetric, define the \emph{symmetric} distribution $p$ on $\{0,1\}^P$ by $p(x) = (1/n!) \sum_{g | g \text{ is a bijection}} t^g(x)$ for all $x \in \{0,1\}^P$.
We call $p$ the \emph{symmetrization} of $t$.
If $t$ is symmetric, then $p = t$.
In other words, shuffling creates symmetry.
For symmetric distributions this shuffling has no effect.

\subsection{Characterizing symmetric distributions}
% \subsubsection{Symmetric marginals}
First we record a straightforward lemma.
Roughly speaking, it says that all same-size marginals of a symmetric distribution agree.

\begin{lemma}\label{lemma:symmetricmarginals}
  Suppose $p: \{0,1\}^P \to [0, 1]$ is a distribution.
  Then $p$ is symmetric if and only if $p_H = (p_J)^g$ for all bijections $g: J \to H$, where $H, J \subset P$.
\end{lemma}

\extraproof{
\begin{proof}
For the \emph{only if} direction, suppose $p$ is symmetric.
Let  $g: J \to H$ be a bijection.
Extend $g$ to a bijection $\bar{g}: P \to P$.
Suppose $u \in \{0,1\}^H$.
Since $\bar g$ is a permutation, use the symmetry of $p$ to conclude that  $p = p^{\bar g}$ and hence
\[
p_H(u)
=
\sum_{z \mid z_{|H} = u} p(z)
=
\sum_{z \mid z_{|H} = u} p^{\bar g}(z)
\]
Now we claim
\[
\sum_{z \mid z_{|H} = u} p^{\bar g}(z)
=
\sum_{z \mid z_{|J} = u \circ g} p(z)
\]
This holds because $p^{\bar{g}}(z) = p(z \circ \bar{g})$ and
\begin{align*}
  \{z \circ \bar g \mid z_{|H} =u \} 
  &= \{y \mid y(j) = (u \circ g)(j) \text{ for all } j \in J
  \}\\
  &= \{
  y \mid y_{|J} = u \circ g
  \}
\end{align*}
% \begin{align*}
%   \{ p^{\bar g}(z) \mid z_{|H} = u\}
%   &= \{p(z \circ \bar g) \mid z_{|H} =u 
%   \} \\
%   &= \{p(y) \mid y(j) = (u \circ g)(j) \text{ for all } j \in J
%   \}\\
%   &= \{
%   p(y) \mid y_{|J} = u \circ g
%   \}
% \end{align*}
Therefore we have
\[
  p_H(u) 
= p_J(u \circ g) 
= (p_J)^g(u)
    \]
Since this holds for any $u$ we have $p_H = (p_J)^g$ as desired.
    The converse holds by choosing $J$ and $H$ equal to $P$.
    \end{proof}
}
An immediate consequence of \Cref{lemma:symmetricmarginals} is that every marginal of a symmetric distribution is symmetric.
To see this, take $H = J$ in the \emph{only if} direction.
This corresponds to the statement that every subset of a set of exchangeable random variables is exchangeable.

\subsubsection{Representation via marginals}
\label{subsubsection:representationviamarginals}

We now look at a specific representation of symmetric distributions,
in terms of a function $q$ such that $q(h)$ is the \textit{marginal}
probability that \textit{any} group of size $h$ tests negative, for $h
= 0, \dots, n$.

\begin{theorem}\label{theorem:zeromarginals}
  Suppose $p: \set{0,1}^P \to [0,1]$ is a distribution.
  Then $p$ is symmetric if and only if there exists a function $q: \set{0, \dots, n} \to [0, 1]$ such that
  \begin{equation}\label{eq:q}
    p_H(\0) = q(\num{H}) \quad \text{for all } H \subset P
  \end{equation}
  We take the convention $p_{\varnothing}(\0) = q(0) = 1$.
\end{theorem}

\begin{proof}
  First, we address the \emph{only if} direction.
  The existence of $q$ is equivalent to the statement that $p_H(\0) = p_{J}(\0)$ whenever $\num{H} = \num{J}$, where $H, J \subset P$.
  As a result, the \textit{only if} direction follows directly from \Cref{lemma:symmetricmarginals} since, using any bijection $g: J \to H$, we have
  \[
    p_H(\0) = (p_J)^g(\0) = p_J(\0 \circ g) = p_J(\0)
  \]
For the \emph{if} direction,  first recall from \Cref{rmk:nnz} that $p$ is symmetric if and only if there exists a function $w$ such that
  \begin{equation}\label{eq:w}
    p(x) = w(\nnz(x)) \quad \text{for all } x \in \{0,1\}^P
  \end{equation}
  Second, we claim that for \emph{any} distribution $p$ on $\{0,1\}^P$ and set $H \subset P$ with $\num{H} = h$
  \begin{equation}\label{eq:marginalzeroexpansion}
    \textstyle
    p_{H}(\0) = \sum_{z|z_{|H} = \0} p(z) = \sum_{i = 0}^{n-h}
    \bbl( \; \sum_{z \mid z_{|H} = \0 \text{ and } \nnz(z) = i} p(z) \; \bbr)
  \end{equation}
  This holds by rearranging the terms in the first sum and grouping them
  according to the number of nonzero entries.

  Suppose by hypothesis that there exists a function $q$ satisfying \cref{eq:q}.
  We will use $q$ to construct a function $w$ satisfying \cref{eq:w} via a linear recursion.
  First define $w(0) = q(n)$.
  Then, for $k = 1, \dots, n$, recursively define 
  \[
    \textstyle
    w(k) = q(n-k) - \sum_{i = 0}^{k - 1}{k \choose i} w(i)
  \]
  We claim that $w$ so constructed satisfies \cref{eq:w}.
  We will show this via strong induction on $k = \nnz(x)$, the number of nonzero values of the outcome $x$.

  First, we introduce a bit of notation which we use to rewrite \cref{eq:marginalzeroexpansion} in terms of $q$.
  Given $x$ in $\{0,1\}^P$ define $I_x = x^{-1}(0)$.
  We claim that
   for any such $x$, we have
  \begin{equation}\label{eq:marginalzeroexpansion3}
    \textstyle
    q(n-\nnz(x)) = \sum_{i = 0}^{\nnz(x)}
    \bbl( \; \sum_{z \mid z_{|I_x} = \0 \text{ and }\nnz(z) = i} p(z) \; \bbr)
  \end{equation}
  This holds by taking $H = I_x$ in \cref{eq:marginalzeroexpansion}, recognizing $n - \num{I_x} = \nnz(x)$, and recognizing $p_{I_x}(\0) = q(n-\nnz(x))$.
  If $\nnz(x) > 0$, we can rearrange \cref{eq:marginalzeroexpansion3} to give
  \begin{equation}\label{eq:rearrangement}
    \textstyle
    p(x) = q(n - \nnz(x)) - \sum_{i = 0}^{\nnz(x)-1} 
    \bbl( \; \sum_{z \mid z_{|I_x} = \0 \text{ and }\nnz(z) = i} p(z) \; \bbr)
  \end{equation}
  This holds because the last term in the outermost sum of \cref{eq:marginalzeroexpansion3} is a sum which itself consists of only one term, specifically
  $\Set{z}{z_{|I_x} = \0 \text{ and } \nnz(z) = \nnz(x)} = \set{x}$.
  For the base case of the induction, $\nnz(x) = 0$, we have 
  $w(0) = p(x)$ for all $x$ with $\nnz(x) = 0$.
  This holds because
  $\nnz(x) = 0$ if and only if $x = \0$, and $q(n) = p_{P}(\0) = p(\0)$.

  Next, suppose the induction hypothesis,
  that for all $z$ with $\nnz(z) < k$, we have $p(z) = w(\nnz(z))$.
  Then for any $x$ in $\{0,1\}^P$ with $\nnz(x) = k > 0$, we have
  \begin{equation}\label{eq:rearrangement2}
    \textstyle
    p(x) = q(n-k) - \sum_{i =0}^{k-1} {k \choose i} w(i)
  \end{equation}
  To see this, take $\nnz(x) = k$ in \cref{eq:rearrangement} and observe that for each of $i = 0, \dots, k-1$, the set
  $\Set{z}{z_{|I_x} = \0 \text{ and } \nnz(x) = i}$
  has ${k \choose i}$ members.
  Each element of the $i$th set has probability $w(\nnz(z))$ under the induction hypothesis. Since $k=\nnz(x)$, \cref{eq:rearrangement2} gives $p(x)$ as a function purely of $\nnz(x)$.
  Consequently, $w$ satisfies \cref{eq:w} as desired.
\end{proof}

\extraclarification{ % for length considerations, we omit this for now
  We reiterate that the functions $q$ of \cref{theorem:zeromarginals} and $w$ of \Cref{rmk:nnz} are two \textit{representations} for a symmetric distribution $p$.
  Either one can readily be computed from the other.
}

We might consider more general sets $D$ containing $0$ other than $\set{0,1}$.
Our definition of symmetry generalizes naturally to distributions on $D^P$ and analogs of \Cref{lemma:symmetricmarginals} and the \emph{only if} direction of \Cref{theorem:zeromarginals} hold.
The \emph{if} direction of \Cref{theorem:zeromarginals} fails, however, if $\num{D} > 2$.
There, $q$ is not a representation.

\subsection{Modeling with symmetric distributions}
We focus here on modeling exchangeable random variables by estimating the parameters of their distribution from empirical data via the principle of maximum likelihood.
We apply this technique in \Cref{section:experiment} below.

We emphasize, however, that modeling is a complex problem with many approaches.
Such approaches include, for example, Bayesian analysis and latent variables.
The choice of approach often depends on available side information.
For a nice recent article, see \cite{niepert2014exchangeable}.

\subsubsection{Maximum likelihood estimation}
\label{subsubsection:fitting}
As is well known, a distribution minimizing the Kullback-Leibler divergence with respect to the \emph{empirical distribution} of a dataset also maximizes the \emph{likelihood} of the dataset.
The following proposition says that one best approximates a distribution in the Kullback-Leibler sense by \textit{symmetrizing} it.
See \Cref{subsubsection:shuffling}.
For notational convenience, denote the equivalence class of $x \in \set{0,1}^P$ by $[x]$.
See \Cref{rmk:nnz}.

\begin{proposition}\label{prop:dklapprox}
  Suppose $r: \{0,1\}^P \to [0,1]$ is a distribution and define the distribution $p^\star: \{0,1\}^P \to [0,1]$ by $p^\star(x) := (1/n!) \sum_{g | g \text{ is a bijection}} r^g(x) = (1/\num{[x]}) \sum_{z \in [x]} r(z)$.
  Then $d_{kl}(r, p^\star) \leq d_{kl}(r, s)$ for all symmetric distributions $s: \set{0,1}^P \to [0,1]$.
\end{proposition}

For this result and others, see Pavlichin \cite{pavlichin2015nearest}.
The order of the arguments of $d_{kl}$ matters.
With the order here, $p^\star$ is called the \textit{M-projection} of $r$ onto the set of symmetric distributions.

One can interpret $p^\star$ of \Cref{prop:dklapprox} as evenly distributing the total probability mass $r$ assigns to each equivalence class among the members of that class.
When the distribution being approximated is the empirical distribution of a dataset, we can easily compute $p^\star$ by counting the number of samples in each of the $n+1$ equivalence classes.
This gives $p^\star(x) = (1/{n \choose \nnz(x)}) \sum_{z | \nnz(z) = \nnz(x)} \frac{1}{m} \num{\Set{i \in \set{1, \dots, m}}{z^i = z}}$ for a dataset $z^1, \dots, z^m$ in $\set{0,1}^P$.

\section{Symmetric and additive set partitioning problems}
\label{section:problems}
A set partitioning problem simplifies considerably if its cost is symmetric and additive.
In this case, it reduces to an \textit{integer} partition problem which can be solved efficiently by any of several methods.
\Cref{prob:mintests} has an additive cost. 
The cost is also symmetric when the distribution is symmetric.

\subsection{Symmetric cost gives an integer partition problem}
\label{subsubsection:setintegerreduction}

Call a function $J: \Disj(P) \to \R$ \emph{symmetric} if $J(F) = J(\presup{g}{F})$ for all $F \in \Disj(P)$ and bijections $g$ on $P$.
For such cost functions $J$, all multiplicity equivalent members of $\Disj(P)$ have the same cost.

\begin{lemma}\label{lemma:integerequivalence}
  Suppose $J: \Disj(P) \to \R$.
  \extraproofalternative{ % don't need label, or displayed, if proof is omitted
    Then $J$ is symmetric if and only if
    $\mu_F = \mu_{G} \implies J(F) = J(G)$ for all $F, G \in \Disj(P)$.
    Here $\mu_F$, $\mu_G$ are the multiplicity functions of $F$ and $G$.
  }
  \extraforproof{ % need the label for reference if proof is included
  Then $J$ is symmetric if and only if
  \begin{equation}\label{eq:integerequivalence}
    \mu_F = \mu_{G} \implies J(F) = J(G) \quad \text{for all } F, G \in \Disj(P) 
  \end{equation}
  Here $\mu_F$, $\mu_G$ are the multiplicity functions of $F$ and $G$.
  }
\end{lemma}
\extraproof{
\begin{proof}
  First, we address the \textit{if} direction.
  Suppose $J$ satisfies \cref{eq:integerequivalence}. 
  Let $F \in \Disj(P)$ and let $g: P \to P$ be a bijection.
  We claim that $F$ and $\presup{g}{F}$ have the same multiplicity function.
  This holds because the set $\Set{g(i)}{i \in H}$ has the same size as $H$ for any $H \in F$.
  By assumption on $J$, $\mu_F = \mu_{\presup{g}{F}}$ implies $J(F) = J(\presup{g}{F})$.
  Since this holds for any $F$ and $g$, we conclude that $J$ is symmetric.

  For the \textit{only if} direction, suppose $J$ is symmetric. 
  Let $F, G \in \Disj(P)$ with $\mu_F = \mu_G$.
  Equivalently, $\lambda_F = \lambda_G$.
  Then $F$ and $G$ have the same size. 
  Denote this number by $p$.

  We begin by constructing an appropriate bijection on $P$.
  Number $F = \{F_1, \dots, F_p\}$ and $G = \{G_1, \dots, G_p\}$ so that 
  \[ 
    \num{F_i} \geq \num{F_j} 
    \text{ and } 
    \num{G_i} \geq \num{G_j} 
    \text{ whenever } 
    1 \leq i < j \leq p
  \]
  We claim that $\num{F_i} = \num{G_i}$ for all $i = 1, \dots, p$.
  This holds because $\lambda_F = \lambda_{G}$.
  Consequently, there are $p$ bijections $h_i: F_i \to G_i$ for $i = 1, \dots, p$.
  Define $\tilde{g}: \cup F \to \cup G$ so that
  \[
    \tilde{g}_{|F_i} = h_i \quad \text{for all } i = 1, \dots, p
  \]
  Then $\tilde{g}$ is a bijection.
  Next, let $g$ be any bijection on $P$ whose restriction on $\cup F$ is $\tilde{g}$.
  By construction, $G = \presup{g}{F}$. 
  Hence $J(G) = J(\presup{g}{F})$, and so by symmetry $J(G) = J(F)$.
\end{proof}
}

\subsubsection{The induced integer partition problem}
\label{subsubsection:inducedintegerpartitionproblem}

Suppose $J: \Disj(P) \to \R$ is \textit{symmetric}.
Then \Cref{lemma:integerequivalence} says that there is a function $J_{\MF}: \MultFns(1, \dots, n) \to \R$ satisfying
\begin{equation}\label{eq:JMF}
  J(F) = J_{\MF}(\mu_F) \quad \text{for all } F \in \Disj(P)
\end{equation}
Moreover, a \emph{partition} $G$ of $P$ minimizes $J$ among all \emph{partitions} if and only if its multiplicity function $\mu_G$ minimizes the restriction of $J_{\MF}$ to $\MultFns(n)$.
Given such a minimizer $\mu^\star$, it is easy to construct a partition of $P$ whose multiplicity function is $\mu^\star$.
Hence, we can find a class of multiplicity equivalent set partitions optimal under $J$ by solving the following problem.

\begin{problem}\label{prob:integerpartition}
  Given $K: \MultFns(n) \to \R$, find multiplicity function $\mu$ to minimize $K(\mu)$.
\end{problem}

We call \Cref{prob:integerpartition} an \textit{integer} partition problem because $\MultFns(n)$ is in one-to-one correspondence with $\IntParts(n)$.
In the next section, we study additional structure on $J$, inherited by $J_{\MF}$, that enables one to avoid exhaustive enumeration in solving \Cref{prob:integerpartition}.

\subsection{Symmetric and additive cost gives an additive integer partition problem}
As usual, we call a function $J: \Disj(P) \to \R$ \textit{additive} if $J(F \cup G) = J(F) + J(G)$ for all disjoint $F, G \in \Disj(P)$.
\extraforproof{ % not needed, since we omit the proof of lemma:symmetricadditivecharacterization
This condition is equivalent to the existence of a function $j: \pow{P} \to \R$ so that
\begin{equation}\label{eq:additivecharacterization}
  \textstyle
  J(F) = \sum_{H \in F} j(H) \quad \text{for all } F \in \Disj(P)
\end{equation}
The function $j$ is related to the function $J$ by $j(H) = J(\set{H})$ for all $H \subset P$.

}
A function $J: \Disj(P) \to \R$ can be symmetric but not additive, and vice versa.
When $J$ is both symmetric and additive, we have the following characterization.

\begin{lemma}\label{lemma:symmetricadditivecharacterization}
  Suppose $J: \Disj(P) \to \R$.
  Then $J$ is symmetric and additive if and only if there exists a function $h: \set{1, \dots, n} \to \R$ so that 
  \extraproofalternative{ % don't need label, or displayed, if proof is omitted
  $J(F) = \sum_{H \in F} h(\num{H})$ for all $F \in \Disj(P)$.
  }
  \extraforproof{ % need the label for reference if proof is included
  \begin{equation}\label{eq:symmetricadditivecharacterization}
    \textstyle
  J(F) = \sum_{H \in F} h(\num{H}) \quad \text{for all } F \in \Disj(P)
  \end{equation}
  }
\end{lemma}

\subsubsection{$J_{\MF}$ inherits the additivity of $J$}

If $J$ is both symmetric and additive, then $J_{\MF}$ inherits the additivity of $J$.
We formalize this statement below.

As usual, we call a function $M: \MultFns(1, \dots, n) \to \R$ \textit{additive} if $M(\mu + \nu) = M(\mu) + M(\nu)$ for all multiplicity functions $\mu$ and $\nu$ with $\mu + \nu \in \MultFns(1, \dots, n)$.
This condition is equivalent to the existence of a function $c: \set{1, \dots, n} \to \R$ satisfying
\begin{equation}\label{eq:multiplicityadditivecharacterization}
  \textstyle
  M(\mu) = \sum_{i = 1}^{n} c(i) \mu(i) \quad \text{for all } \mu \in \MultFns(1, \dots, n)
\end{equation}

\begin{lemma}\label{lemma:JMFinheritsadditivity}
  Suppose $J: \Disj(P) \to \R$ is symmetric with $J_{\MF}: \MultFns(1, \dots, n) \to \R$ defined as in \cref{eq:JMF}.
  If $J$ is additive, then $J_{\MF}$ is additive.
\end{lemma}
\extraproof{
\begin{proof}
  Let $\mu$ and $\nu$ be in $\MultFns(1, \dots, n)$ with $\mu + \nu \in \MultFns(1, \dots, n)$.
  Take $F$ and $G$ to be any \emph{disjoint} sets in $\Disj(P)$ so that $\mu$ and $\nu$ are the multiplicity functions of $F$ and $G$, respectively.
  Since $F \cap G = \varnothing$, the multiplicity function of $F \cup G$ is $\mu + \nu$.
  We claim
  \[
    J_{\MF}(\mu + \nu) = J(F \cup G) = J(F) + J(G) = J_{\MF}(\mu) + J_{\MF}(\nu)
  \]
  The first and third relations hold by definition of $J_{\MF}$ and the second holds because $J$ is additive.
  Since this holds for any $\mu$ and $\nu$, we conclude that $J_{\MF}$ is additive.
\end{proof}
}

The characterization in \cref{eq:multiplicityadditivecharacterization} says that any additive function on $\MultFns(1, \dots, n)$ has a representation $c: \set{1, \dots, n} \to \R$.
\Cref{lemma:symmetricadditivecharacterization} says a symmetric and additive function $J$ on $\Disj(P)$ has a representation $h: \set{1, \dots, n} \to \R$.
For the \emph{additive} function $J_{\MF}$, these coincide.
\extraforproof{ % not appropriate, since we omit the proof of lemma:symmetricadditivecharacterization
In fact, one can alternatively show the \textit{if} direction of \Cref{lemma:symmetricadditivecharacterization} above using the result of \Cref{lemma:JMFinheritsadditivity} and the characterization of \cref{eq:multiplicityadditivecharacterization}.
}

\subsubsection{The induced additive integer partition problem}
As before, suppose $J$ is additive and symmetric.
Since $J_{\MF}$ inherits the additivity of $J$, we can find a class of multiplicity equivalent set partitions optimal under $J$ by solving the following problem.
\begin{problem}\label{prob:additiveintegerpartition}
  Given $c: \set{1, \dots, n} \to \R$, find $\mu \in \MultFns(n)$ to minimize $\sum_{i = 1}^{n} c(i) \mu(i)$.
\end{problem}

Similar to \Cref{prob:integerpartition}, the one-to-one correspondence between $\MultFns(n)$ and $\IntParts(n)$ means that we can interpret \Cref{prob:additiveintegerpartition} as finding an \emph{integer} partition.
For this reason, prior work has called \Cref{prob:additiveintegerpartition} an integer partition problem \cite{engel2014optimal}.
We use the terminology \emph{additive} integer partition problem to distinguish \Cref{prob:additiveintegerpartition} from the general \Cref{prob:integerpartition}. 

\subsection{Minimizing tests for symmetric distributions}

Given a \emph{symmetric} distribution of specimen statuses, \Cref{prob:mintests} reduces to an additive integer partition problem.
In this case the objective, which is additive for any distribution, is \emph{also} symmetric.
We formalize this fact in \Cref{cor:mintestssymadd} below, which, given \Cref{lemma:symmetricadditivecharacterization}, is an immediate consequence of the following.

\begin{lemma}\label{lemma:Etestsgroup2}
  Suppose $x$ has distribution $p$.
  If $p$ is symmetric, then there exists a function $U: \set{1,\dots,n} \to \R$ satisfying $\E T_H(x) = U(\num{H})$ for all nonempty $H \subset P$.
\end{lemma}
\begin{proof}
  Let $H \subset P$ be nonempty.
  We make three straightforward substitutions in \cref{eq:expectedtests}.
%   Given nonempty $H \subset P$, we have by definition
%   \begin{equation}\label{eq:expectedtests3}
%     \E T_H(x) = \begin{cases}
%       1 & \text{if } \num{H} = 1 \\
%       1 + \num{H}\Prob(S_H(x) = 1) & \text{otherwise}
%     \end{cases}
%   \end{equation}
%   See \Cref{eq:expectedtests}.
%   We make three straightforward substitutions.
  First, the status of $H$ is either 0 or 1.
  Consequently, $\Prob(S_H(x) = 1) = 1 - \Prob(S_H(x) = 0)$.
  Second, the status of $H$ is 0 if and only if $x_i = 0$ for all $i \in H$.
  Hence, $\Prob(S_H(x) = 0) = p_{H}(\0)$.
  Finally, the symmetry of $p$ is equivalent to the existence of a function $q$ satisfying $p_H(\0) = q(\num{H})$ for all $H \subset P$.
  See \Cref{theorem:zeromarginals}.
  Substituting into \cref{eq:expectedtests} gives
%   Substituting into \Cref{eq:expectedtests3} gives
  \[
    \E T_H(x) = \begin{cases}
      1 & \text{if } \num{H} = 1 \\
      1 + \num{H}(1 - q(\num{H})) & \text{otherwise}
    \end{cases}
  \]
  The right hand side is a function of $\num{H}$, as desired.
\end{proof}

\begin{corollary}\label{cor:mintestssymadd}
  Suppose $x$ has distribution $p$. Define $J: \Disj(P) \to \R$ by $J(F) = \E C(F, x)$ for all $F \in \Disj(P)$.
  Then $J$ is additive. If $p$ is symmetric, then $J$ is symmetric.
\end{corollary}
\extraproof{
\begin{proof}
  $J$ is additive by definition.
  See \cref{eq:totalexpectedtests}.
  Now suppose $p$ is symmetric.

  By the \textit{if} direction of \Cref{lemma:symmetricadditivecharacterization}, it suffices to exhibit a function $h: \set{1, \dots, n} \to \R$ satisfying \cref{eq:symmetricadditivecharacterization}.
  We take $h$ to be the function $U$ of \Cref{lemma:Etestsgroup2}.
  We claim
  \[
    J(F) = \sum_{H \in F} U(\num{H}) \quad \text{for all } F \in \Disj(P)
  \]
  We conclude that $J$ is additive and symmetric, as desired.
\end{proof}
}

\subsection{Solving additive integer partition problems}
\label{subsection:solutions}
Several efficient approaches for computing the solutions of \emph{additive} integer partition problems are known \cite{engel2014optimal, onn2015some}.
% In this section we briefly mention these before elaborating on a dynamic programming method.

% \subsubsection{Overview of approaches}
% \label{subsubsection:approaches}

% We highlight three approaches to \Cref{prob:additiveintegerpartition}.
% These include a polyhedral formulation, a dynamic programming method, and a reduction to a minimum-weight path problem in a directed graph.

\paragraph{Linear programming} A polyhedral approach identifies multiplicity functions with vectors in $\mathbb{R}^n$.
The convex hull of these \textit{multiplicity vectors}, called the \textit{integer partition polytope}, has a succinct polyhedral lift \cite{onn2015some}. 
Linear optimization over the extended formulation can be done via linear programming and a solution recovered via projection.
See \cite{onn2015some} for details.

\paragraph{Dynamic programming} Alternatively, a dynamic programming approach minimizes, in order from $k = 1, \dots, n$, the cost $c$ over the set $\MultFns(k)$.
A solution for the $k+1$ case is found using the solutions for the cases $1, \dots, k$.
See \Cref{subsubsection:dpapproach} below and \cite{engel2014optimal} for details.

\paragraph{Shortest path problem} Finally, a minimum-weight path reduction constructs a directed graph in which the multiplicity vectors correspond to the directed paths between two distinguished vertices.
By appropriately weighting the edges of these paths, the minimizing multiplicity vectors are put in one-to-one correspondence with the minimum-weight paths.
Consequently, one can find a minimizing multiplicity vector by solving the well-known shortest weighted path problem via standard algorithms.
See \cite{onn2015some} for details.

\subsubsection{A dynamic programming approach}
\label{subsubsection:dpapproach}

Here we expand on the dynamic programming algorithm mentioned above.
% in \Cref{subsubsection:approaches}.
Further details and variants are given in \cite{engel2014optimal}.
% For further details and also a variant of \Cref{prob:additiveintegerpartition} that seeks an optimal integer partition with \textit{fewest} parts, see \cite{engel2014optimal}.

The algorithm we describe sequentially computes optimal partitions of all integers $k \leq n$ in order from $k = 1, \dots, n$.
At step $k$, it uses the costs of optimal partitions of $1, \dots, k-1$ to find an optimal partition of $k$.
As usual, we call a partition of an integer $k$ \emph{optimal} if its multiplicity function minimizes the objective of \Cref{prob:additiveintegerpartition} over $\MultFns(k)$.

\paragraph{Interpretation}

Since we omit proofs below, we start by interpreting the algorithm.
We imagine partitioning the integer $n$ by first choosing to include a part of size $i \leq n$, and subsequently partitioning the remainder $n - i$.
It is easy to see that every partition of $n$ can be obtained in this way.
If the cost is additive, then the cost of such a partition is the cost of the part $i$ plus the cost of the partition chosen for $n - i$.
Given a fixed $i$, we can minimize this cost by optimally partitioning $n - i$.
So if we knew in advance the cost of optimally partitioning each integer smaller than $n$, then we could optimize over our choice of the first part $i$.
The same interpretation applies to partitioning $n-i$, and so on, recursively.

The algorithm proceeds in reverse of this interpretation. 
First we optimally partition $1$, then $2$, then $3$, and so on up to $n$.
To illustrate concretely, first we optimally partition 1.
This is trivial, since there is only one choice of partition.
Next, we optimally partition 2 by either taking a part of size $1$, inducing the partition $1+1$, or keeping the single part $2$.
Likewise for $3$. 
We may take a part of size $1$ and use our optimal partition of $2$, take a part of size $2$ and use our optimal partition of $1$, or take a single part of size $3$.
Which is best depends on the cost of partitioning $2$. 
We have already computed this value at the previous step.
Similarly for 4. 
We take a part of size $1$, $2$, $3$ or $4$. 
The choice depends on the cost of optimally partitioning $1$, $2$ and $3$, which we have already computed.
We continue in a similar way up to $n$.

% We ignore here the minor subtlety that we can skip considering a part of size $i$ if its cost exceeds the cost of an optimal partition of $i$.
% In this case, an optimal partition of $k \geq i$ will not include a part of size $i$ since this part could be replaced to lower the cost.

\paragraph{Subproblem Optimal Value Recursion}
We briefly formalize this interpretation.
Given a function $M: \MultFns(1, \dots, n) \to \R$, define the function $\Mstar: \set{0, \dots, n} \to \R$
by $\Mstar(0) = 0$ and
\begin{equation}\label{eq:Mstar}
  \Mstar(k) = \min\Set{M(\mu)}{\mu \in \MultFns(k)} \quad \text{for } k = 1, \dots, n
\end{equation}
$\Mstar$ is called the \emph{value function}.
$\Mstar(k)$ is the cost of an optimal partition of $k$.
If $M$ is additive, then $\Mstar$ satisfies the following recursive relation.
See Theorem 1 of \cite{engel2014optimal}.

\begin{lemma}\label{lemma:recursion}
  Suppose $M: \MultFns(1, \dots, n) \to \R$ is additive with representation $c: \set{1, \dots, n} \to \R$ satisfying \cref{eq:multiplicityadditivecharacterization}.
  Then $\Mstar$ defined as in \cref{eq:Mstar} satisfies
  \[
    \Mstar(k) = \min\Set{\Mstar(k - i) + c(i)}{i \in \set{1, \dots, k}} \quad \text{for all } k \in \set{1, \dots, n}
  \]
\end{lemma}
Hence we can use $\Mstar(1), \dots, \Mstar(k-1)$ to compute $\Mstar(k)$.

\paragraph{Algorithm}
\Cref{lemma:recursion} justifies a simple algorithm for computing $\Mstar(1), \dots , \Mstar(n)$ and corresponding multiplicity functions $\muiter{1}, \dots , \muiter{n}$ satisfying $M(\muiter{k}) = \Mstar(k)$ for $k = 1, \dots, n$.
In other words, $\muiter{k}$ is the multiplicity function of an optimal partition of $k$.
We let $\muiter{0}$ be the constant zero function for notational convenience.

We iterate from $k = 1, \dots, n$.
At step $k$, we find an integer $\iiter{k}$ so that
\[
  \iiter{k} \in \argmin \Set{\Mstar(k-i) + c(i)}{i \in \set{1, \dots, k}}
\]
Using $\iiter{k}$ and $\muiter{k - \iiter{k}}$, we define the multiplicity function $\muiter{k} \in \MultFns(k)$ by
\[
  \muiter{k}(j) = \begin{cases}
  \muiter{(k - \iiter{k})}(j) + 1 & \text{if } j = \iiter{k} \\
  \muiter{(k - \iiter{k})}(j)     & \text{otherwise }
\end{cases}
\]
We can interpret $\muiter{k}$ as an extension of $\muiter{k-\iiter{k}}$ which includes one additional part of size $\iiter{k}$.
We choose the part $\iiter{k}$ to minimize the sum of its cost and the cost of optimally partitioning $k - \iiter{k}$.
By construction, $\muiter{k}$ has cost $\Mstar(k) = \Mstar(k-\iiter{k}) + c(\iiter{k})$ and so is optimal.
See \Cref{lemma:recursion}.
In particular, the multiplicity function $\muiter{n}$ corresponds to an optimal partition of $n$.

This algorithm has quadratic time complexity.
% In other words, it requires a number of real arithmetic and comparison operations which grows quadratically in $n$.
% To see this, notice that there are $n$ steps of the algorithm and at step $k$ we minimize over a finite set of size $k$.
%
% We also mention that \Cref{prob:additiveintegerpartition} and \Cref{lemma:recursion} have straightforward variants in which we restrict the support of the multiplicity function.
We also mention that \Cref{prob:additiveintegerpartition} and \Cref{lemma:recursion} have analogs in which we restrict the support of the multiplicity function.
% This corresponds to restricting the sizes of the parts of the integer partition.
% In the context of group testing, for example, an upper bound on the size of the groups may be motivated by the testing capability.

\section{Numerical example on empirical data}
\label{section:experiment}
In this section we apply the tools of symmetric probability and group testing to an empirical dataset from the COVID-19 pandemic.
The example is meant to illustrate several approaches.
It is not intended to improve upon testing methodology used in a practical setting at this point.

We start by describing the origin and preparation of the dataset.
Then we compare four pooling strategies.
The first three strategies happen to coincide for this dataset whereas the final one, using tools developed in this paper, gives a different and more efficient pooling.

\subsection{Dataset background}

The Hebrew University-Hadassah COVID-19 Diagnosis Team provide the pooled testing data we use below \cite{barak2021lessons}.

\paragraph{Testing context}
The COVID-19 pandemic called for large-scale and high-throughput disease screening.
 Authorities encouraged specimen pooling to conserve test resources \cite{fda2020coronavirus}.

The Hebrew University team processed 133,816 nasopharyngeal lysates across 17,945 pools via Dorfman screening between April 19 and September 16, 2020 
\cite{barak2021lessons}.
They collected these specimens from \textit{asymptomatic} individuals and performed tests for screening purposes.
Their protocol adaptively switched between size-5 and size-8 pools.

\paragraph{Testing pipeline}
Specimens arrived in \textit{batches}, often of size 80.
Technicians centrifuged each lysate before a robot performed pooling and mixing.
Up to 92 pooled or individual specimens could be tested simultaneously in a single \textit{run} of a reverse transcription polymerase chain reaction (PCR) machine.
The pool size and specimen-to-pool assignment were informed by the prior week's prevalence and batch-specific side information.

\paragraph{Correlated specimen statuses}
The team observed empirical efficiency \emph{exceeding} that indicated by Dorfman's analysis.
For example, at a prevalence of 1.695\% the size-8 pools enjoyed an \emph{empirical} efficiency of 4.59 whereas Dorfman predicts a \emph{theoretical} efficiency of 3.96.

They attribute this discrepancy to the ``nonrandom distribution of positive specimens in pools."
They report that ``specimens arrive in batches: from colleges, nursing homes, or health care personnel."
Technicians sorted related specimens into pools ``such that family members and roommates were often pooled together, thereby increasing the number of positive samples within the pool."
This protocol \emph{helps} efficiency because keeping positive specimens together mitigates the number of positive pools and, hence, retests required.

These circumstances challenge the typical probabilistic assumption that specimen statuses are independent.
For this particular dataset, therefore, modeling statuses as exchangeable may be more appropriate than modeling them as independent.

\newcommand{\adjustedNumberOfSpecimens}{112,848}
\newcommand{\adjustedNumberOfPools}{14,106}
\newcommand{\batchedNumberOfBatches}{1,410}
\newcommand{\batchedNumberOfSpecimens}{112,800}
\newcommand{\batchedNumberOfPools}{14,100}

\subsection{Dataset preparation}
\label{subsection:dataset_preparation}
We simplify the dataset before using the PCR machine \textit{run timestamp} to impute size-80 batches of specimens.

\paragraph{Simplifications}

We ignore (a) pools without a timestamp (b) pools of size 5 (c) pools with specimens of inconclusive status.
The first measure allows us to batch sequentially; the second mitigates the varying prevalence rate across pool sizes; the third ensures we have complete status data.
These adjustments leave \adjustedNumberOfSpecimens{} specimens across \adjustedNumberOfPools{} pools.

\paragraph{Batching}
Although the dataset does not include information about which samples arrived together, it does include information about \textit{when} pools were tested in the PCR machine.
We use this \textit{run timestamp} to order the samples and impute batches of size 80 sequentially.

This protocol yields \batchedNumberOfBatches{}  batches including \batchedNumberOfSpecimens{} specimens across \batchedNumberOfPools{} pools.
One could alternatively batch within a particular day, PCR run, or by using different batch sizes (e.g., 40 or 64).
Our experiments indicate that these results closely correspond to the size-80 sequential batching case and so we do not include details here.

\newcommand{\numberOfShufflingRepetitions}{10,000}

\subsection{Experimental setup}
We compare four strategies to pool the finite population of size 80.
The last is enabled by the tools of this paper.
The strategies are:
\begin{enumerate}
  \item
    \textit{Hebrew University team}.
    Use 10 size-8 pools \cite{barak2021lessons}.
  \item
    \textit{Dorfman}.
    Use the pool size indicated by Dorfman's infinite population analysis \cite{dorfman1943detection}. 
    Include one extra smaller pool if this size does not evenly divide 80.
  \item
    \textit{Independent statuses}.
    Select a pooling to minimize the expected tests used under an estimated \emph{IID} distribution.
    Use the algorithm of Hwang \cite{hwang1975generalized} or of \Cref{subsubsection:dpapproach}.
  \item
    \textit{Exchangeable statuses}.
    Select a pooling to minimize the expected tests used under an estimated \emph{symmetric} distribution.
    Use the algorithm of \Cref{subsubsection:dpapproach}.
\end{enumerate}

Strategy (1) requires no estimation, strategies (2) and (3) require estimating the population prevalence, and strategy (4) requires estimating the parameters of a symmetric distribution.
For (4) we use the principle of maximum likelihood (see \Cref{subsubsection:fitting}).

Each strategy may indicate a distinct pooling.
We evaluate these under the estimated symmetric distribution and against the empirical data.
We report the empirical efficiencies both \textit{with} and \textit{without} randomization over specimen-to-pool assignment.
For the former we randomize over \numberOfShufflingRepetitions{} trials.
We also report the theoretical efficiency of size-8 pools as indicated by Dorfman's infinite population analysis and under the estimated finite IID model.

We emphasize that these strategies do not model the dynamics of the underlying infection.
Also, we leave to future work a comparison with other models that capture correlation differently, such as those considering community structure.
See \Cref{subsection:correlations}.

\newcommand{\expEmpiricalPrevalence}{1.624}
\newcommand{\expTheoreticalDorfmanEfficiency}{4.04}
\newcommand{\expTheoreticalIIDEfficiency}{4.04}
\newcommand{\expTheoreticalSizeEightEfficiency}{4.38}
\newcommand{\expTheoreticalSizeTenEfficiency}{4.48}
\newcommand{\expEmpiricalSizeEightEfficiencyNoShuffle}{4.71}
\newcommand{\expEmpiricalSizeTenEfficiencyNoShuffle}{4.75}
\newcommand{\expEmpiricalSizeEightEfficiencyShuffled}{4.38}
\newcommand{\expEmpiricalSizeTenEfficiencyShuffled}{4.48}
\newcommand{\expEmpiricalSizeEightEfficiencyShuffledStdError}{0.02}
\newcommand{\expEmpiricalSizeTenEfficiencyShuffledStdError}{0.02}

\subsection{Experimental results}
% \paragraph{Distribution fitting}
The empirical prevalence is \expEmpiricalPrevalence\%.
We visualize the estimated IID and symmetric distributions used for strategies (3) and (4) in \Cref{figure:models_analysis}.

\begin{figure}[htbp]\label{figure:models_analysis}
  \centering
  \includegraphics[width=\textwidth]{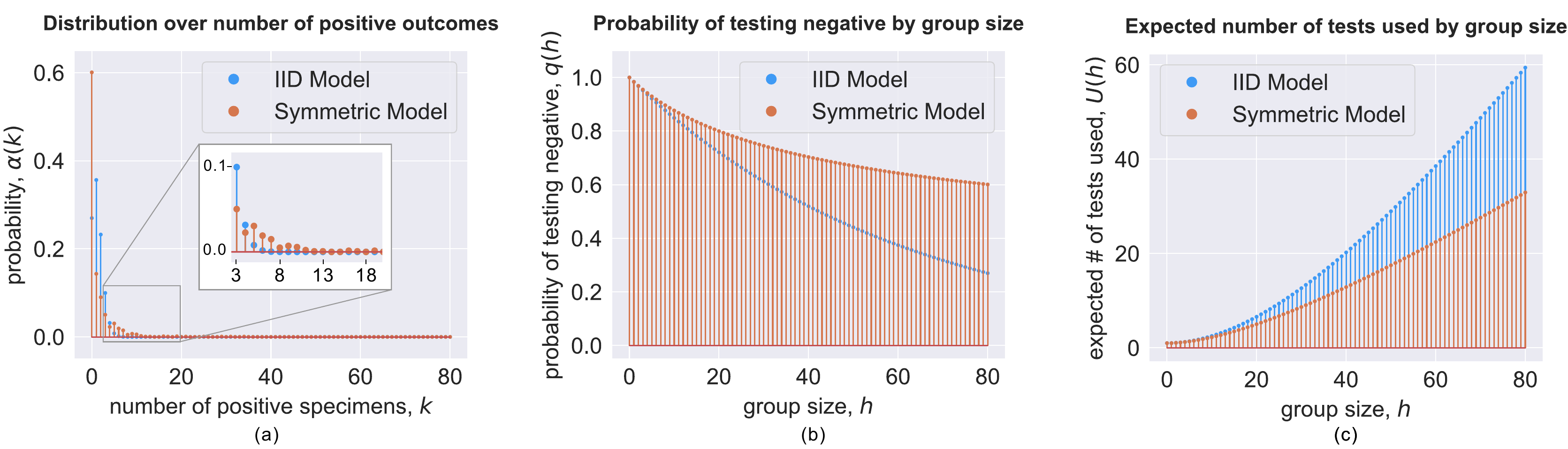}
  \vspace{-2.0em}
  \caption{
  Comparison of an independent and identically distributed (IID) model with a symmetric model for a population of size 80. 
  (a) The representation $\alpha$ of these distributions where $\alpha(k)$ is the probability of seeing $k$ positive specimens (see \Cref{proposition:finitefinetti}). 
  The IID model decays more rapidly.
  The symmetric distribution has non-monotonic decay; e.g., it assigns more mass to five positives than four positives.
  (b) The representation $q$ where $q(h)$ is the probability that a group of size $h$ tests negative (see \Cref{theorem:zeromarginals}).
  The IID model underestimates these probabilities.
  (c) The function $U$ where $U(h)$ is the expected number of tests used on a group of size $h$ (see \Cref{lemma:Etestsgroup2}).
  The IID model overestimates these costs.
  }
\end{figure}

\paragraph{Pooling strategies} 
Strategies (1), (2) and (3) each indicate 10 pools of size 8 whereas strategy (4) indicates 8 pools of size 10.
Since the strategies only indicate two distinct options, we refer to \textit{size-8 pools} and \textit{size-10 pools} in the discussion below.

\paragraph{Theoretical efficiencies} 
Both Dorfman's infinite analysis and the finite IID model indicate an efficiency of \expTheoreticalDorfmanEfficiency{} for size-8 pools.
Under the estimated symmetric model, the efficiency of the size-8 and size-10 pools is \expTheoreticalSizeEightEfficiency{} and \expTheoreticalSizeTenEfficiency{}, respectively.

\paragraph{Empirical efficiencies}
The average empirical efficiencies of the size-8 and size-10 pools are \expEmpiricalSizeEightEfficiencyShuffled{} and \expEmpiricalSizeTenEfficiencyShuffled{}, respectively.
The standard error in both cases is \expEmpiricalSizeEightEfficiencyShuffledStdError{}.
Without randomization, the empirical efficiencies for size-8 and size-10 pools are \expEmpiricalSizeEightEfficiencyNoShuffle{} and \expEmpiricalSizeTenEfficiencyNoShuffle{}, respectively.

\subsection{Discussion and interpretation}

% \paragraph{Distributions}
The estimated IID and symmetric distributions differ visibly (see \Cref{figure:models_analysis}).
The IID distribution \emph{underestimates} the probabilities of observing $\geq 6$ positive specimens (see \Cref{figure:models_analysis}, panel a).
Hence, it also the \emph{underestimates} probability that a group of a particular size will test negative (see \Cref{figure:models_analysis}, panel b).
As a result, it \emph{overestimates} the number of tests used for a group of a particular size (see \Cref{figure:models_analysis}, panel c).
For simplicity, we do not discuss the important topic of uncertainty in these estimates.

\paragraph{Strategies}
The first three strategies \emph{coincide} whereas strategy (4) uses \emph{fewer} and \emph{larger} pools.
In our experience, it is usual for the symmetric model to indicate larger pools.
We attribute this phenomenon to the overestimation described in the previous paragraph.

We emphasize that strategies (1), (2) and (3) need \emph{not} coincide.
When they do differ, in our experience, it is often by indicating successively larger pools.
For example, strategy (3) avoids the small remainder pool indicated by (2) when Dorfman's pool size does not evenly divide the population size.
Here the strategies agree.
They also agree with the Hebrew University team's original choice of size-8 pools.

\paragraph{Theoretical efficiencies}
The efficiencies indicated by Dorfman and the finite IID model (a) agree, (b) exceed that reported in \cite{barak2021lessons} and (c) \emph{underestimate} the theoretical efficiency as predicted by the symmetric distribution.
Phenomenon (a) may be interpreted to justify Dorfman's approximation.
Phenomenon (b) occurs because the prevalence is \emph{slightly} lower in our processed data than the original dataset.
Phenomenon (c) is a consequence of the IID model overestimating the expected number of tests used.

Under the estimated symmetric distribution, the efficiency of the size-10 pools exceeds that of the size-8 pools.
We expect the size-10 pools to be at least as efficient as the size-8 pools as a consequence of the optimization carried out in strategy (4).

\paragraph{Empirical efficiencies}
With randomization, the mean empirical efficiencies agree with their theoretical values.
The standard errors are relatively small.

Without randomization, both size-8 and size-10 efficiencies increase with the size-10 efficiency remaining larger.
This increase appears to be a consequence of the intentional pooling carried out by the Hebrew University Team.
Since we batch and pool sequentially, the size-8 pools used here match exactly those constructed by the team.
Although the size-10 efficiency is higher here, our experience indicates that this difference is not significant.
The empirical efficiency reported in \cite{barak2021lessons} is \emph{lower} than these values because certain pools were retested even though each of the pool's specimens was negative.

\section{Prior work on group testing}
\label{section:prior_work}
Here, we elaborate on the areas and applications of group testing.
We highlight, in particular, the emerging area of group testing under correlations.

\subsection{Areas of group testing}
We start by outlining various divisions of group testing.
% Although a comprehensive survey of group testing is beyond the scope of this paper, we highlight some of the variety within the field.

% MODEL, PROBABILISTIC VS COMBINATORIAL
\subsubsection{Specimen status models, side information, and objective}
A first distinction in group testing is between the probabilistic and combinatorial approach.
In this paper, we consider the \emph{probabilistic} approach.
Here, one specifies a probabilistic model of specimen statuses and performs testing to minimize some criterion usually related to expected efficiency.
In the alternative \emph{combinatorial} approach, one specifies information about the number of positive specimens and performs testing to minimize a worst-case criterion usually related to the maximum number of tests required.
Hence a worst-case, or \emph{minimax} \cite{du2000combinatorial}, analysis replaces an average-case analysis.
For examples of combinatorial group testing, starting with Li in 1962
\cite{
li1962sequential}, see
\cite{
hwang1972method,
katona1973combinatorial,
chang1980group,
chang1981group,
gilstein1985optimal,
eppstein2007improved}.
Hereafter we assume the probabilistic setting.

% CHOICE OF MODEL
Within probabilistic group testing, a \emph{first further} subdivision involves the related aspects of model choice and side information.
For example, Dorfman \cite{dorfman1943detection} models binary specimen statuses as IID and assumes only knowledge of the \emph{prevalence}, i.e. the probability that any given specimen is positive.
Historically, several authors have followed his approach
\cite{
sterrett1957detection,
sobel1959group,
ungar1960cutoff,
finucan1964blood,
watson1961study,
gurnow1965note,
sobel1971symmetric,
graff1972group,
samuels1978exact,
turner1988calculus}, including two influential textbooks containing the so-called \emph{blood testing problem} as exercises \cite{feller1968introduction,wilks1962mathematical}.
This simple probabilistic model has been called the \emph{IID model} \cite{aldridge2019group}, \emph{binomial model} \cite{pfeifer1978dorfman-type,chen1990using}, \emph{B-model} \cite{sobel1968binomial}, or \emph{binomial} \cite{kumar1971finding} or \emph{homogeneous} population \cite{aprahamian2019optimal}.
Some authors use the term \emph{binomial group testing} \cite{sobel1968binomial,hwang1974finding}.

% BEYOND THE BINOMIAL
Many other probabilistic models have been considered besides the binomial.
In 1968, Sobel considered a setting in which $d$ positive specimens are distributed uniformly throughout the population \cite{sobel1968binomial}.
This paradigm is called the \emph{hypergeometric model}
\cite{
sobel1966binomial1,
sobel1968binomial,
hwang1981hypergeometric}, \emph{H-model} \cite{sobel1968binomial}, or \emph{combinatorial prior}
\cite{
aldridge2019group}.
The term \emph{generalized hypergeometric model} \cite{hwang1978note} has been used when only an upper bound on $d$ is assumed, whereas the term \emph{truncated binomial model} \cite{hwang1980optimal,hwang1984robust} has been used  if an upper bound is known for the binomial model.
In 1973, Nebenzahl and Sobel \cite{nebenzahl1973finite} considered group testing for a population composed of several separate binomial subpopulations each with a different prevalence.
As we mention above, Hwang \cite{hwang1975generalized} further generalized this direction by modeling specimens as independent but \emph{nonidentical} binary random variables.
This paradigm is called the \emph{generalized binomial model} \cite{hwang1975generalized}, \emph{prior defectivity model} \cite{aldridge2019group}, \emph{nonidentical model} \cite{li2014group,kealy2014capacity,doger2021group}, or \emph{heterogeneous population} \cite{elhajj2022screening,aprahamian2019optimal,black2012group}.

% EXCHANGEABLE MODEL AND ITS RELATION
In this style, we might use the term \emph{exchangeable model} or \emph{symmetric model} to describe the exchangeable populations we consider herein.
The information assumed is the $n$ parameters of the symmetric distribution.
On one hand, the binomial, truncated binomial, hypergeometric, and generalized hypergeometric models are symmetric.
On the other, the generalized binomial model is \emph{not} symmetric.
Previously, the so-called \emph{mean model} \cite{hwang1984robust} has been studied as a generalization of all of these.
It assumes only the mean number of positives.
Later on, we discuss other more recent models allowing correlation between specimen statuses.

% PARAMETER UNCERTAINTY & OBJECTIVE
Within probabilistic group testing, a \emph{second further} subdivision relates to parameter uncertainty and the choice of objective.
Starting with Sobel and Groll in 1959 \cite{sobel1959group}, several authors handle uncertainty in model parameters
\cite{
sobel1966binomial2,
kumar1973asymptotically,
chen1990using,
johnson1999dual}.
In this setting, the objective of \emph{estimation} may replace that of efficiency
\cite{
sobel1975group,
chen1990using}.
Elsewhere, other objectives such as information gain \cite{abraham2020bloom} and risk-based metrics \cite{aprahamian2019optimal} have been considered.
See \cite{hitt2019objective} for further discussion of different objectives.
In this paper, we assume full knowledge of model parameters and focus on the objective of efficiency as measured by the expected number of tests used.

\subsubsection{Testing models, feasibility, and noise}
A second distinction relates to the group testing capability.
Dorfman \cite{dorfman1943detection} considers unconstrained, noiseless, binary individual and group tests.
This is the setting we consider.
The terms \emph{reliable} \cite{du2000combinatorial} for noiseless and \emph{disjunctive} \cite{berger2002asymptotic} for binary group outcomes are also used.
We mention alternative testing models for binary specimen statuses below.
Historically, other testing models also arise naturally from nonbinary specimen status models, as for example the \emph{trinomial model} \cite{kumar1970group-testing} and \emph{multinomial model} \cite{kumar1970multinomial}.

Toward \emph{more} informative tests, we mention three examples. 
First, Sobel \cite{sobel1968binomial} considered \emph{quantitative group testing} \cite{aldridge2019group} in which group tests reveal the number of positive specimens.
Sobel used the term \emph{H-type} model, in contrast with the term \emph{B-type} model for the usual binary result setting.
As indicated earlier, he used analogous language for the specimen model.
Other terms include \emph{linear model} or \emph{adder channel model} \cite{aldridge2019group}.
For variants on this theme, see \cite{mehravari1986generalized,amin2014semiquantitative,wang2023tropical}.
Second, Pfeifer and Enis \cite{pfeifer1978dorfman-type} considered group tests that reveal the sum or mean of individual test results.
Although the distinction involves tests and not statuses, they use the term \emph{modified binomial model} or \emph{M-model}.
For examples of fully continuous test results, see \cite{wein1996pooled,wang2018group}. 
Third, Sobel and coauthors \cite{sobel1971symmetric} considered \emph{symmetric group testing} \cite{du2000combinatorial} in which there are three group test outcomes: all positive, all negative, and mixed.
We reiterate that symmetric here describes the test model and not the status model.

In the opposite direction, several authors weaken the group test capability.
Toward \emph{less} informative models, we mention Hwang \cite{hwang1976group} and Farach et al. \cite{farach1997group} who consider dilution effects and so-called \emph{inhibitor} specimens, respectively.
For details and other examples, see the survey \cite{aldridge2019group} and the book \cite{du2000combinatorial}.
Similarly, application areas often motivate various forms of \emph{constrained group testing} \cite{aldridge2019group}.
Two natural and classic examples limit the size of a group test \cite{hwang1975generalized} or the divisibility of a specimen \cite{sobel1959group}.
Recently, these have been studied under the heading \emph{sparse group testing} \cite{gandikota2016nearly,gandikota2019nearly,
inan2017sparse,
inan2020sparse}.
For a second example, in \emph{graph-constrained group testing} the tests must correspond to paths in a given graph \cite{harvey2007non-adaptive,karbasi2012sequential,cheraghchi2012graph-constrained,sihag2021adaptive,sihag2021adaptive,spang2019unconstraining}.
The methodology we give herein readily handles limits on group size.
We consider no additional constraints.

Starting with Graff and Roeloffs in 1972 \cite{graff1972group}, authors regularly study probabilistic models of noisy, or \emph{unreliable}, tests \cite{kim2007comparison,johnson1999dual,ahn2023adaptive}.
Noisy tests motivate studying \emph{non-exact}, or \emph{partial}, recovery as opposed to \emph{exact} recovery \cite{aldridge2019group}.
In this paper, we only consider reliable tests.

\subsubsection{Algorithms and analysis}
A third distinction in group testing involves the algorithms and analysis considered.
The algorithmic distinction is largely captured by a division into \emph{adaptive} and \emph{nonadaptive} procedures.
The analytical distinction is largely captured by a division into \emph{finite population} and \emph{infinite population}, or \emph{asymptotic}, analysis.

The terms \emph{adaptive}, \emph{sequential} and \emph{multistage} describe procedures with multiple \emph{rounds}, \emph{cycles}, or \emph{stages} of testing \cite{du2000combinatorial}.
The tests of later rounds may depend on, and so adapt to, the results of earlier ones.
Each round may involve one test or several.
The literature is replete with adaptive algorithms
\cite{
lee1972dorfman,
gill1974identification,
hwang1975generalized,
kealy2014capacity,
broder2020note,
doger2021group}.
The further modifiers \emph{nested} \cite{hwang1976optimum,malinovsky2019revisiting} and \emph{hierarchical} \cite{kim2007comparison,tebbs2013two-stage,johnson1991inspection} indicate that groups tested in later stages are subsets of groups already tested.
For example, a classic multistage nested approach is Sobel and Groll's original \emph{binary splitting} or \emph{halving} \cite{sobel1959group}.
For a modern discussion and further examples, see \cite{aldridge2019group, du2000combinatorial}.

Alternatively, various applications motivate \emph{nonadaptive}, or \emph{single-stage}, algorithms in which all group tests must be specified in advance
\cite{
hwang1987non-adaptive,
balding1996optimal,
chan2011non-adaptive,
cheraghchi2011group,
dyachkov2014lectures,
cheraghchi2012graph-constrained,
aldridge2019group}.
Although it is sometimes natural in this case to discuss the two stages of testing and \emph{decoding} \cite{aldridge2019group}, we use the term \emph{two-stage} exclusively in its usual sense \cite{berger2002asymptotic,debonis2005optimal,mezard2011group} of two rounds of testing.

Dorfman's \cite{dorfman1943detection} particular two-stage, adaptive strategy splits the population into non-overlapping groups of a fixed size, tests these, and individually retests the specimens of positive groups.
The strategy has been called the \emph{Dorfman procedure}
\cite{
graff1972group,
hwang1975generalized,
pfeifer1978dorfman-type,
hwang1984robust},
\emph{Dorfman-type group testing} \cite{pfeifer1978dorfman-type},
\emph{Dorfman screening}
\cite{
mcmahan2012informative,rewley2020specimen},
\emph{Dorfman testing }
\cite{aprahamian2019optimal},
and \emph{single pooling} \cite{broder2020note}.
Some authors use the terms \emph{conservative} \cite{aldridge2022conservative} or \emph{trivial} \cite{debonis2005optimal} when the second round of a two-stage procedure only involves individual retests.
These terms are usually employed, however, when \emph{confirmatory} \cite{farach1997group} tests are used to verify suspected positives indicated by a first round of \emph{overlapping} tests.
This occurs, for example, in \emph{array testing} \cite{phatarfod1994use,mcmahan2011two-dimensional}.

Dorfman \cite{dorfman1943detection} considers the setting in which the population size tends to infinity.
This asymptotic regime remains popular, especially in the information theory community \cite{aldridge2019group}.
On the other hand, starting with Sobel and Groll in 1959 \cite{sobel1959group}, many authors consider finite populations
\cite{
kumar1970multinomial,
gill1974identification,
hwang1975generalized} or both settings
\cite{
nebenzahl1973finite,
sobel1971symmetric,
pfeifer1978dorfman-type}.
We consider herein the finite-population setting in which Dorfman's procedure is generalized slightly to allow for groups of different sizes.
Hence one seeks a partition of the population. 
See \Cref{section:problem} for details.

\subsection{Applications of group testing}
Next, we discuss applications.

\subsubsection{Beyond disease screening}
In his original article, Dorfman \cite{dorfman1943detection} speculated on the utility of group testing outside of medical testing.
In particular, he mentioned manufacturing quality control.
Sobel and Groll's influential 1959 article \cite{sobel1959group} gave further examples. 
See also the book \cite{johnson1991inspection}.
Since then, researchers have applied group testing techniques in such diverse settings as wireless communications 
\cite{
hayes1978adaptive,
berger1984random,
wolf1985born,
luo2008neighbor,
inan2017sparse,
inan2018energy-limited,
inan2019group,
inan2020sparse,
cohen2020efficient}, genetics 
\cite{
green1990systematic,
bruno1995efficient,
macula1998probabilistic,
hwang2006pooling},
machine learning \cite{ubaru2017multilabel,zhou2014parallel,malioutov2013exact,liang2021neural},
signal processing \cite{gilbert2008group,cohen2021serial} and data stream analysis \cite{cormode2005whats,amid2014poisson}.
See the survey \cite{aldridge2019group} and book \cite{du2000combinatorial} for further applications and references.

\subsubsection{The COVID-19 pandemic}
The COVID-19 pandemic created a surge of interest in group testing for disease screening \cite{
mallapaty2020mathematical,
ellenberg2020five,
austin2020pooling}.
We make four observations.
First, pooling was feasible.
Standard technology detects SARS-CoV-2 virus in pools of up to 32 specimens \cite{yelin2020evaluation}.
For more on test sensitivity and optimal pool sizes in real-world scenarios, see
\cite{
wang2021performance,
bateman2021assessing,
hanel2020boosting,
burtniak2023dorfman}.
Second, pooling was widely and successfully used in practical settings \cite{hogan2020sample,yelin2020evaluation,lohse2020pooling,ben-ami2020large-scale,barak2021lessons} and encouraged by authorities \cite{cdc2021interim,fda2020coronavirus,augenblick2020group,sunjaya2020pooled,abdalhamid2020assessment,verdun2021group,daniel2021pooled}.
Third, practitioners often preferred Dorfman's procedure for reasons, among simplicity, that we detail below \cite{ben-ami2020large-scale,barak2021lessons}.
Other sophisticated approaches were, however, proposed \cite{mutesa2021pooled,ghosh2021compressed,hong2022group}.
Finally, the classical independence theory failed to explain empirical findings in large-scale asymptomatic screening \cite{barak2021lessons,comess2022statistical}.
We discuss this phenomenon and work aiming to remediate it below.

\subsubsection{Benefits of Dorfman's procedure}
There are reasons to prefer Dorfman's protocol beyond its simplicity, historical precedence and modern importance.
First, it divides each specimen into only two aliquots.
This feature is relevant when the testing process is destructive or dilutive, as is usually the case in disease screening or any biological specimen testing.
Second, it is parallel.
Within both stages, all indicated tests can be performed at the same time.
Consequently, the latency is predictable and bounded.
The test efficiency gains of more sophisticated procedures, e.g. Sterrett's \cite{sterrett1957detection} or binary splitting \cite{sobel1959group}, are often offset by latency considerations.
Third, it has easy to compute pool sizes and interpretable results.
The methodology we develop for exchangeable populations also enjoys these features.

\subsection{Group testing with specimen status correlation}
\label{subsection:correlations}
Finally, the COVID-19 pandemic created a surge of interest in studying group testing under models motivated by infectious disease screening.
These often include \emph{correlation} between statuses, a feature largely absent from the classical literature.
Furthermore, these models involve various forms and degrees of side information.
While it is reasonable to suppose that such side information can help efficiency, it is natural to be interested in methodology independent of it.
Modeling exchangeability requires no additional knowledge of contact tracing, interaction networks, or community structure.

Although Hwang mentions correlated statuses in his 1984 discussion of the mean model \cite{hwang1984robust}, Lendle et al.'s 2012 article \cite{lendle2012group} appears to be the first to study correlated specimen statuses in earnest.
They investigate a restricted form of exchangeability and show that efficiency gains can result from pooling \emph{within} clusters of positively correlated specimens.

We highlight three more recent directions.
First, Lin et al. \cite{lin2021positively} study Dorfman testing for a correlated arrival process of contiguous groups from different IID populations.
They report higher efficiency.
% They also propose a hierarchical method for the case in which a social graph is available.
Other simulation \cite{rewley2020specimen,deckert2020simulation} and theoretical \cite{wan2022correlation} investigations also report that pooling within positively correlated groups increases efficiency.
Second, Ahn et al.  \cite{ahn2021adaptive,ahn2023adaptive} study a so-called \emph{stochastic block infection model}.
They analyze a modified binary splitting algorithm which uses knowledge of a specimen's community membership.
For other generalizations of the IID model related to theirs, see \cite{gonen2022group,lau2022model-based,nikpey2023group}.
Third, and related, Nikolopoulos et al. \cite{nikolopoulos2021group,nikolopoulos2023community-aware} study \emph{community-aware group testing}, in which a hypergraph encoding overlapping communities is known, and propose algorithms leveraging this side information.
% They similarly propose algorithms leveraging this side information.

These examples are characteristic of a growing body of work incorporating information such as cluster identity \cite{arasli2021graph,attia2021heterogeneity,baccini2021pool,best2023efficient}, an underlying network topology \cite{bertolotti2021network,bertolotti2022inference,sewell2022leveraging}, or contact-tracing \cite{goenka2021contact,tsirtsi2022pooled,cao2023group} into models.
Also, several authors study disease spread and so consider \emph{dynamic} models \cite{srinivasavaradhan2021entropy,bi2022approximate,doger2022dynamical,srinivasavaradhan2022dynamic,arasli2022group,arasli2023dynamic,arasli2023group}.
Prior to the pandemic, side information was usually incorporated via specimen-specific probabilities of testing positive \cite{hwang1975generalized,bilder2010informative,mcmahan2012informative}.

Finally, we mention that Comess et al. \cite{comess2022statistical} also investigate the unexpectedly high efficiency observed by Barak et al. \cite{barak2021lessons}, the source of the data we consider in \Cref{section:experiment}.
They propose and analyze a community network model.
They use this model to also investigate the higher-than-expected \emph{sensitivity} observed by Barak et al. \cite{barak2021lessons}.
We do not consider sensitivity herein.

\section{Conclusion}
\label{section:conclusion}
In this paper, we develop and apply tools for Dorfman's two-stage adaptive group testing protocol.
In particular, we study the problem under the modeling assumption that the statuses are exchangeable and so their distribution is symmetric.

This modeling assumption is both amenable to analysis and relevant for infectious disease screening.
Although symmetric distributions are a simple model of reality, they nonetheless allow for correlation among specimen statuses.
Such correlations appear in disease screening because specimens originating from the same family, living space, or workplace often arrive for testing, and hence for pooling, together.
Since the disease is contagious, positive statuses co-occur.
Accounting for this phenomenon in the probabilistic model may indicate better efficiency and larger pool sizes than proposed by the classical theory.
The dataset we studied in \Cref{section:experiment} exhibits this feature.
In summary, symmetric distributions are a prototypical class on the path to further research into and analysis of more complicated models.

\subsection{Future directions}
We focus on topics related to Dorfman's procedure.
It may also be of interest, however, to analyze other group testing protocols like Sterrett's procedure \cite{sterrett1957detection} or Sobel and Groll's binary splitting \cite{sobel1959group} under exchangeability.
Also, exchangeability might be used to tighten bounds on the number of tests required in the nonadaptive setting \cite{atia2012boolean}.
% Also, exchangeability might improve known bounds \cite{atia2012boolean} on the number of tests required in the nonadaptive setting.
% It would also be interesting to establish tighter bounds on the number of tests required in the nonadaptive case than those which are known for the IID setting. See \cite{atia2012boolean}.

\paragraph{Notable variants and generalization}
We list three variants and a generalization of the symmetry considered in this paper.
The three variants are (1) \textit{infinite} population exchangeability, (2) \emph{test error} models for exchangeable statuses, and (3) \emph{risk-adjusted objectives} incorporating, e.g., the variance of the number of tests used.
Even within the finite population, error-free, minimize-expected-tests setting of this paper, an interesting generalization of this paper may study distributions which are invariant under an \emph{arbitrary} permutation group.

\paragraph{Characterizing savings and robustness}
How much can we save by correctly modeling the statuses as exchangeable instead of independent?
Toward answering this, suppose $x$ has symmetric distribution $p$ and denote by $\bar{p}$ the IID distribution whose prevalence matches that of $p$.
Suppose $G^\star$ and $\bar{G}^\star$ are corresponding optimal partitions under $p$ and $\bar{p}$, respectively.
One approach to the question of savings is to study the quantity $\Delta(p) := \E C(\bar{G}^\star, x) - \E C(G^\star, x)$.
What is $\sup_p \Delta(p)$? Which symmetric distributions achieve this value?
There are no savings if $\bar{G}^\star = G^\star$, but \Cref{section:experiment} indicates that distributions with savings exist and appear empirically.

Also, how robust are these approaches to uncertainty in estimated parameters?
Given an interval containing the population prevalence or a set containing the symmetric distribution, what are the optimal \emph{worst-case} partitions?
The linear programming approach (see \Cref{subsection:solutions}) may be useful for these questions and the foregoing one.

\paragraph{Using features to estimate the probability a group tests negative}
Lastly, we sketch a direction toward more complicated distributions.
Although specimens have identical marginals under the symmetric models considered in this paper, it is natural to relax this assumption as well.
Classically, Hwang \cite{hwang1975generalized} proposed using specimen-specific negative-status probabilities.
He showed that, assuming independence, one can efficiently compute partitions to minimize the expected number of tests used.
With modern tools, one might use \emph{features} and \emph{logistic regression} to estimate these probabilities.
See \cite{bilder2010informative} for an approach along these lines.

To generalize, one may drop the independence assumption and directly estimate the probability that a \emph{group} tests negative by, for example, performing logistic regression on sets of individual specimen features.
Regression models that do not depend on the order of an input list of feature vectors are called \emph{permutation-invariant} \cite{zaheer2017deep,bloem-reddy2020probabilistic}.
Given such a model indicating the probability that a group tests negative, one might then employ general-purpose partitioning algorithms to find partitions which minimize the expected number of tests used.

\clearpage

% \appendix
% \section{An example appendix}
% \section*{Acknowledgments}

\bibliographystyle{siamplain}
\bibliography{references}

\end{document}

\maketitle

\section{A detailed example}

Here we include some equations and theorem-like environments to show
how these are labeled in a supplement and can be referenced from the
main text.
Consider the following equation:
\begin{equation}
  \label{eq:suppa}
  a^2 + b^2 = c^2.
\end{equation}
You can also reference equations such as \cref{eq:matrices,eq:bb}
from the main article in this supplement.

\lipsum[100-101]

\begin{theorem}
An example theorem.
\end{theorem}

\lipsum[102]

\begin{lemma}
An example lemma.
\end{lemma}

\lipsum[103-105]

Here is an example citation: \cite{KoMa14}.

\section[Proof of Thm]{Proof of \cref{thm:bigthm}}
\label{sec:proof}

\lipsum[106-112]

\section{Additional experimental results}
\Cref{tab:foo} shows additional
supporting evidence.

\begin{table}[htbp]
\footnotesize
  \caption{Example table.}  \label{tab:smfoo}
\begin{center}
  \begin{tabular}{|c|c|c|} \hline
   Species & \bf Mean & \bf Std.~Dev. \\ \hline
    1 & 3.4 & 1.2 \\
    2 & 5.4 & 0.6 \\ \hline
  \end{tabular}
\end{center}
\end{table}

\bibliographystyle{siamplain}
\bibliography{references}